\documentclass[final,leqno,onefignum,onetabnum]{siamltex1213}
\usepackage[utf8]{inputenc}
\usepackage{subfigure,
				showkeys,
				amsmath,
				amssymb,
            color,
            comment	
				}
\usepackage{graphicx,dcolumn}				
\usepackage[usenames,dvipsnames]{xcolor}
\usepackage[font=small]{caption} 

\title{Epidemic Outbreaks in Networks with Equitable or Almost-Equitable Partitions \thanks{This work has been partially supported by the European Commission within the framework of the CONGAS project FP7-ICT-2011-8-317672 (see \url{http://www.congas-project.eu})}} 

\newtheorem{remark}{Remark}

\def\lcm{\mathrm{ lcm}}

\def\bP{\mathbb P}

\def\eps{\varepsilon}

\def\eps{\varepsilon}

\def\diag{\mathop{\rm diag}}

\def\E#1{{\mathrm E}\left [{#1} \right ]}

\def\hB{\widehat{B}}

\definecolor{hanpurple}{rgb}{0.32, 0.09, 0.98}
\definecolor{electricultramarine}{rgb}{0.25, 0.0, 1.0}
\newcommand{\so}[1]{\textcolor{Black}{#1}}
\newcommand{\rev}[1]{\textcolor{Black}{#1}}
\newcommand{\fdp}[1]{\textcolor{Black}{#1}}
\newcommand{\ovp}{{\overline{p}}}

\author{S. Bonaccorsi\footnotemark[1]
\and S. Ottaviano\footnotemark[1]\ \footnotemark[2]
\and D. Mugnolo\footnotemark[3]
\and F. De Pellegrini \footnotemark[2]}

\begin{document}
\maketitle

\renewcommand{\thefootnote}{\fnsymbol{footnote}}

\footnotetext[1]{Mathematics Department, University of Trento,via Sommarive 14, 38123 Povo (Trento), Italy}
\footnotetext[2]{CREATE-NET, Via alla Cascata 56/d, 38123 Povo (Trento), Italy}
\footnotetext[3]{Institute of Analysis, University of Hagen, 58084 Hagen, Germany}
\slugger{mms}{xxxx}{xx}{x}{x--x}

\begin{abstract}
We study the diffusion of epidemics on networks that are partitioned into local
communities. The gross structure of hierarchical networks of this kind can be described by a \emph{quotient graph}. 
The rationale of this approach is that individuals infect 
\so{those belonging to the same community with higher probability than individuals in other communities.} In community models the nodal infection probability is thus expected to depend mainly on the interaction of a few, large interconnected clusters.
In this work, we describe the epidemic process as a continuous-time \so{individual-based} 
\rev{susceptible--infected--susceptible} (SIS) model \so{using  a first-order mean-field approximation. }

 \rev{A key feature of our model is that} the spectral radius of this smaller quotient graph (which only captures the macroscopic structure of the community network) is all we need to know in order to decide whether the overall healthy-state defines a \so{globally asymptotically stable} or an unstable equilibrium. \so{Indeed, the spectral radius is related to the epidemic threshold of the system.}

\so{Moreover we prove that, above the threshold, another steady-state exists that can be computed  using a lower-dimensional dynamical system} associated with the evolution of the process on the quotient graph. Our investigations are based on the graph-theoretical notion of \emph{equitable partition} and of its recent and rather flexible generalization, that of \emph{almost equitable partition}.
\end{abstract}

\begin{keywords}susceptible-infected-susceptible model, hierarchical networks, graph spectra, equitable and almost \so{ equitable} partitions\end{keywords}
\begin{AMS}\end{AMS}

\pagestyle{myheadings}
\thispagestyle{plain}
\markboth{Epidemic Outbreaks in Networks with Equitable or Almost-Equitable Partitions}{S. Bonaccorsi, S. Ottaviano, D. Mugnolo, F. De Pellegrini}

\section{Introduction}

\rev{Metapopulation models of epidemics consider the entire population partitioned into communities (also called households, clusters or subgraphs). Such models assume that each community shares a common environment or is defined by a specific relationship (see, e.g., ~\cite{Hanski, Masuda2010, Allen2007}).}
%

\so{Several authors also account for the effect of migration between communities \cite{Colizza2008,Poletto2013}. Conversely, the model we are interested in suits better the diffusion of computer 
viruses or stable social communities, which do not change during the infection period; hence we do not consider migration.}

In this work, we study the diffusion of epidemics over an undirected graph $G=(V,E)$ with edge set $E$ and node set $V$. The order of $G$, denoted $N$, is the cardinality of $V$, whereas the size of $G$ is the cardinality of $E$, denoted $L$. Connectivity of the graph $G$ is conveniently encoded in the $N \times N$ adjacency matrix $A$.
We are interested in the case of networks that can be naturally partitioned into $n$ communities: they are represented by a node set partition $\pi=\left\{V_1,...,V_n\right\}$, i.e., a sequence of mutually disjoint nonempty subsets of $V$, called cells, whose union is $V$. 

The epidemic model adopted in the rest of the paper is a continuous-time Markovian individual-based \rev{susceptible--infected--susceptible} (SIS) model. In the SIS model a node can be repeatedly infected, recover and yet be infected again. The viral state of a node $i$, at time $t$, is thus described 
 by a Bernoulli random variable $X_i(t)$, where we set $X_i(t) = 0$ if $i$ is healthy and $X_i(t) = 1$ if $i$ is infected. 
\so{Every node at time $t$ is either infected with probability $p_i(t) = \bP(X_i(t) = 1)$ or healthy (but susceptible) with probability $ 1 - p_i(t)=\bP(X_i(t) = 0) $}. Each node becomes infected following a Poisson process with rate $\beta$. Also, $i$ recovers following a Poisson process with rate  $\delta$. We further assume that infection and curing processes are independent~\cite{VanMieghem2009}. The ratio $\tau=\beta/\delta$ is called the \textit{effective spreading rate}.

Recently, also non Markovian types of epidemic spread were introduced in the literature, by choosing other than the exponential interaction time for infection and/or curing (see \fdp{for instance} \cite{NonM, genInf}). However, this seems to complicate the analysis considerably and \fdp{it is beyond the scope of this work}. 

Compared to the homogeneous case where the infection rate is the same for all pairs of nodes, in our model we consider two infection rates: the {\em intra-community} infection rate $\beta$ for infecting individuals in the same community and the {\em inter-community} infection rate $\eps\beta$  i.e., the rate at which individuals
among different communities get infected. \so{ We assume  $0<\eps < 1$, the customary physical interpretation being that infection across communities occur at a much smaller rate. Clearly the model can be extended to the case $\eps \geq 1$.} 

 \rev{ Further models where the epidemic process within communities is faster compared to the rate at which it spreads across communities, have  been studied in literature \cite{Bonaccorsi, Ball1997, Ball2008, Ross2010}.}

As described in~\cite{VanMieghem2009,VanMieghem2012b}, the SIS process developing on a graph with $N$ nodes 
is modeled as a continuous-time Markov process with $2^N$ states. The dynamics of the nodal infection 
probability is obtained by the corresponding Kolmogorov differential equations, but the resulting dynamical system consists of $2^N$ linear differential equations, not a viable approach for large networks. \so{Hence, often, an approximation of the SIS process is needed}. In this work we consider the first-order mean-field approximation NIMFA, proposed by Van Mieghem et al. in~\cite{VanMieghem2009, VanMieghem2012a, VanMieghem2011}. 

NIMFA replaces the original $2^N$ linear differential equations by $N$ non-linear differential equations; they represent the time-change of the infection probability of a node. As typical in first-order approximations of SIS dynamics, the only approximation required by NIMFA is that the infectious state of two nodes in the 
 network are uncorrelated, i.e., $\E{X_i(t)X_j(t)}= \E{X_i(t)}\E{X_j(t)}$. 

\subsection{Long-term prediction and epidemic threshold}\label{Epidemic Threshold}

For a network with finite order $N$, the exact SIS Markov process will always converge towards its \so{unique} absorbing state, \so{that is the zero-state where all nodes are healthy}. 
\so{The other states 
form a transient class, from which one can reach the zero-state with positive probability.} \fdp{ Because transitions from the zero-state have zero probability\footnote{Some models, as, e.g., the $\eps$-SIS model  \cite{VanMieghem2012b}, include the possibility of a nodal self-infection, thus making the whole process irreducible.} the stochastic model predicts that the virus will disappear from the network \cite{Pollett90}.
 }

\so{However the waiting time to absorption is a random variable whose distribution depends on the initial state of the system, and on the parameters of the model \cite{ NasselCLosed, Nassell2002}. In fact there is a critical value $\tau_c$ of the effective spreading rate $\tau= \beta/\delta$, whereby if  $\tau > \tau_c$ \rev{the time to absorption  grows exponentially in $N$, while for $\tau < \tau_c$  the infection vanishes exponentially fast in time}. 
  The critical value $\tau_c$ is often called the \textsl{epidemic threshold} \cite{VanMieghem2009, Bailey1975, Daley1999, Pastor2001}. } 

\fdp{Thus above the threshold, a typical realization of the epidemic process \rev{may experience} a very long waiting time before absorption to the zero-state. During such waiting time, the so-called {\em quasi-stationary distribution} can be used in order to approximate the probability distribution of occupancy of the system's states}. 
\so{The quasi-stationary distribution is obtained by conditioning on the fact that there is no extinction \cite{NasselCLosed,Nassell2002}. The quasi-stationary distribution can be regarded as the limiting conditional distribution, useful in representing the long-term behavior of the process {\em ``that in some sense terminates, but appears to be stationary over any reasonable time scale''}\cite{Pollett}.}


\rev{In fact, numerical simulations of SIS processes also reveal that, already for reasonably small networks $(N \geq 100)$ and when $\tau > \tau_c$, the overall-healthy state is only reached after an unrealistically long time. Hence, the indication of the model is that, in the case of real networks, one should expect that the extinction of epidemics is hardly ever attained \cite{VanMieghem2013, Draief2010}.}
For this reason the literature is mainly concerned with establishing the value of the epidemic threshold, being a key parameter behind immunization strategies related to the network protection against viral infection.

For an SIS process on graphs, $\tau_c$ depends on the spectral radius $\lambda_1(A)$ of the adjacency matrix $A$  \cite{Wang2003,VanMieghem2009}. NIMFA determines the epidemic threshold for the effective spreading rate as $\tau^{(1)}_c =\frac{1}{\lambda_{1}(A)}$, where the superscript (1) refers to the first-order mean-field approximation \cite{VanMieghem2009,VanMieghem2014}. \rev{Farther, in Thm. \ref{thresh}, we shall study the asymptotic behavior of the solutions of the NIMFA system, both above and below the threshold.} 

\so{We observe that, with respect to the exact Markovian SIS model, the state of nodes was recently proved to be not negatively correlated ~\cite{Cator_positive_correlations}. It is hence possible to prove that, due to the assumption of independence, NIMFA yields an upper bound for the probability of infection of each node, as well as a lower bound for the epidemic threshold, i.e., $\tau_c = \alpha \tau_c^{(1)}$ with $\alpha \geq 1$. }

\fdp{From the application standpoint, a key issue is to determine for which networks of given order NIMFA performs worst, meaning that $\alpha =\frac{\tau_c}{\tau_c^{(1)}}$ is largest. To this respect, further efforts have been made to satisfactorily quantify the accuracy of the model \cite{Accuracy}.}

\so{Finally, when $\tau > \tau^{(1)}_c$, a limiting occupancy probability appears as the second constant solution\footnote{We remember that all bounded trajectories of an autonomous first-order differential equation tend to an equilibrium, i.e., to a constant solution of the equation.} of the \so{NIMFA} non-linear system which exists, apart from the zero-vector solution. 
 {\rev Such} non-zero steady-state reflects well the observed  viral behavior \cite{VanMieghem2012a}: it can be seen as the {\rev analogous} of the quasi-stationary distribution of the exact stochastic SIS model.}




\begin{figure*}[t]	
	\centering
	\includegraphics[width=0.60\textwidth]{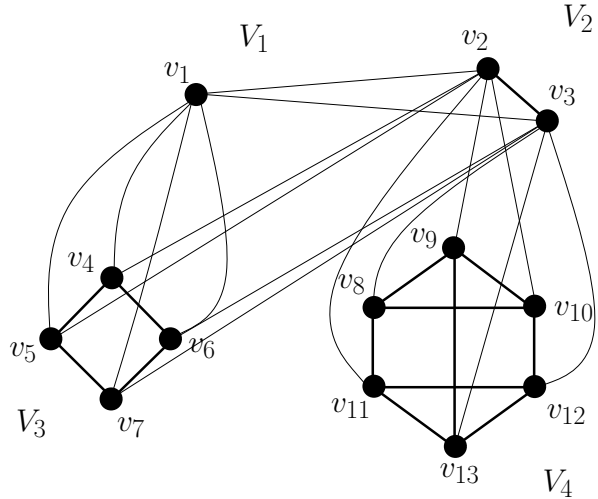}
\caption{A sample graph with equitable partition $V=\{\{v_1\},\{v_2,v_3\},\{v_3,v_4,v_5,v_6\},\{v_7,v_8,v_9,v_{10},v_{11},v_{12},v_{13}\}\}$.} \label{fig:fig11}
\end{figure*}
\subsection{Outline and main results}

 As already observed in \cite{Bonaccorsi}, the presence of communities generates a strong mixing effect at local level (e.g., the rate of infection inside a community tends to be homogeneous) as opposed to the much lower speed of mixing (i.e., much larger inhomogeneity) within the whole population. In \cite{Bonaccorsi} a complete graph represents the internal structure of each community. Such assumption appears natural for small community orders, for example, because the members of a small community usually know each other, as they may be friends, relatives, members of a common club, employees of the same department, etc. Moreover, given two connected communities, all of their nodes are mutually linked.

In this work, instead, 
\rev{we allow for the case of sparser community structures.} More precisely we consider an \textsl{equitable partition} of the graph. First of all this means that all nodes belonging to the same community have the same internal degree: formally the subgraph $G_i$ of $G(V,E)$ induced by $V_i$ is regular for all $i$'s (recall that $\pi=\left\{V_1,...,V_n\right\}$ is a partition of the node set $V$, which is assumed to be given a priori).  
Furthermore, for any two subgraphs $G_i,G_j$, each node 
in $G_i$ \rev{is connected with the same number of nodes in $G_j$}.

 The macroscopic structure of such a network 
can be described by the \emph{quotient graph} $G/\pi$, an oriented graph (possibly) featuring loops and multiple edges. The nodes of the quotient graph are the cells $V_1,\ldots,V_n$ in $G$. \so{In the last part of the work we extend our study to the case of \textit{almost equitable partitions}} that does not require any \so{specific} structural condition inside \rev{each $G_i$}.

Such network structure can be observed, e.g., in the architecture of some computer networks where clusters of clients connect to single routers, whereas the routers' network has a connectivity structure with nodes' degree constrained by the number of ports. Also, graphs representing  multi-layer networks  may be characterized using equitable and almost equitable partitions \cite{moreno2014}.

\so{In Sec. \ref{epid} we \rev{describe} the NIMFA differential equations and provide \rev{an} analysis of the global dynamics that allows us to identify the epidemic threshold $\tau_c^{(1)}$. In Sec.~\ref{sec:equi}, after defining equitable partitions, we introduce the so-called quotient matrix $Q$ that is related to $G/\pi$. \rev{Since} matrix $Q$ has the same spectral radius of adjacency matrix $A$, \rev{a novel expression is found} for the bound on the epidemic threshold $\tau_c^{(1)}$  as a function of network metrics.} Thus, a relation between the epidemic threshold and the spectral properties of the corresponding quotient matrix is obtained. 

\so{In Sec.~\ref{sec:InfDyn} we show under which conditions the matrix $Q$ can be used in order to express the whole epidemic dynamics by a system of $n$ equations instead of $N$, where $n < N$. We prove the existence of a positively invariant  set for the original system of $N$ differential equations that contains the equilibrium points. Moreover we show that, above the threshold, when a second non-zero equilibrium point appears, we can use the reduced system for its computation. }

In Sec.~\ref{AlmEq} we finally extend our investigations to the case of almost equitable partitions. We consider the special case of almost equitable partitions obtained by perturbing an equitable one, i.e., by adding/deleting 
a certain set of edges from an equitable partition. Thus, we relax the assumption that the internal structure of each community is regular. Even in this case we obtain a lower bound for the epidemic threshold.


\section{The epidemic model}\label{epid}


\rev{The NIMFA model describes the process of diffusion of epidemics on a graph by expressing the time-change 
of the probability $p_i$ that node $i$ is infected.}

Thus, node $i$ obeys a following differential equation~\cite{VanMieghem2009}
\begin{equation}\label{A}
\frac{ d p_i(t)}{dt} = (1-p_i(t))\beta\left(\sum_{j=1}^N a_{ij}p_j(t)\right) -\delta p_i(t), \; i=1,\ldots,N .
\end{equation}
In \eqref{A} the time-derivative of the infection probability of node $i$ consists of two competing processes: 
\begin{enumerate}
\item while healthy with probability $1-p_i(t)$, all infected neighbors, whose average number is 
$\sum_{j=1}^N a_{ij}p_j(t)$, infect node $i$ at rate $\beta$. 
\item while node $i$ is infected with probability $p_i(t)$, it is cured at rate $\delta$. 
\end{enumerate}
\rev{The following matrix representation of \eqref{A} holds}
\begin{equation}\label{mat}
\frac{dP(t)}{dt}= \beta AP(t)-\diag(p_i(t))(\beta AP(t)+ \delta u),
\end{equation}
where $P(t)=(\,p_1(t) \, p_2(t)  \dots  p_N(t)\,)^T$, $\operatorname{diag}(p_i(t))$ is the diagonal
matrix with elements $p_1(t), p_2(t),  \dots ,p_N(t)$ and $u$ is the all-one vector. From \eqref{mat}, 
considering $P(t)= \operatorname{diag}(p_i(t))u$, we can write
\begin{eqnarray}\label{mat2}
\frac{dP(t)}{dt}&&= \beta A P(t)-\delta \operatorname{diag}(p_i(t))u - \operatorname{diag}(p_i(t))\beta AP(t)\nonumber \\
                &&=(\beta A-\delta I)P(t) - \beta \operatorname{diag}(p_i(t)) A P(t).
\end{eqnarray}
\so{Clearly we study the system for $(p_1, \dots, p_N) \in I_N= [0,1]^N$. It can be shown that the system \eqref{mat2} is positively invariant in $I_N$, i.e. if $P(0) \in I_N$ then $P(t) \in I_N$ for all $t > 0$ \cite[Lemma ~3.1]{Stab}.}

\so{\rev{The} analysis of the global dynamics of \eqref{mat2} leads to identify the epidemic threshold $\tau^{(1)}_c$ in terms of the effective spreading rate $\tau=\beta/\delta$ where, as mentioned in Sec. \ref{Epidemic Threshold}},
\begin{equation}\label{tauc}
\tau^{(1)}_c = \frac{1}{\lambda_{1}(A)},
\end{equation}
with $\lambda_1(A)$ spectral radius of $A$. This critical value separates the absorbing phase from the endemic phase. We shall prove this, in Thm \ref{thresh}, by studying the stability of the equilibrium points of \eqref{mat2}, that are solutions of the equation
\begin{equation}\label{eqpoints}
P= \frac{\beta}{\delta} (I-\operatorname{diag}(p_i))A P.
\end{equation}
To this aim we shall adapt the results in \cite{Stab} to our individual-based SIS model. Let us denote by $f$ the right hand side of \eqref{mat2}, i.e., \eqref{mat2} can be re-written as a vector-valued differential equation
\begin{equation}\label{f}
\frac{dP}{dt}=f(P),
\end{equation}
where $\displaystyle f:[0,1]^N\rightarrow \mathbb{R}^N$ is a $C^{\infty}$ function.
 Let $P_0=0$ be the vector of all zero components, one can easily check that $P_0$ is an equilibrium point of the system (\ref{f}), i.e. $f(P_0)=0$. Also, the following holds
\begin{theorem}\label{thresh}
 If $\tau \leq 1/\lambda_1(A)$ then $P_0$ is a globally asymptotically stable equilibrium of (\ref{mat2}).\\ 
If $\tau > 1/\lambda_1(A)$, $P_0$ is unstable and there exists another equilibrium point  $P_{\infty}$ 
 that is globally asymptotically stable in $I_N - \left\{0\right\}$. 
\end{theorem}
\begin{proof}
We can rewrite the system  \eqref{f} in the following form (see \cite[p. 108]{DiffEqandDynamicalSystem})
\begin{equation}\label{Df} 
\dot{P}= D_f P + F(P),
\end{equation}
where  $D_f$ is the Jacobian matrix of $f$ at $P_0$ and $F(P)$  is a column vector whose $i$-th component is $-\beta \sum_{j=1}^N a_{ij} p_i p_j$.

From (\ref{mat2}) we have
\begin{equation*}
\left(Df(P_0)\right)_{ij}= \begin{cases}
\beta a_{ij} & i \neq j\\
- \delta & i=j
\end{cases}
\end{equation*}
 that is $D_f=\beta A- \delta I$.
\rev{Since adjacency matrix $A$ is real and symmetric its eigenvalues are real. Hence, the eigenvalues of $D_f$ are real as well and of the form }
\begin{equation*}
\lambda_i(D_f)=\beta\lambda_i(A)-\delta.
\end{equation*} 
In particular, let $\lambda_1(D_f)= \max_{i} \lambda_i(D_f)$, since the spectral radius of $A$ is positive we have 
\begin{equation*}
\lambda_1(D_f)=\beta \lambda_1(A)-\delta.
\end{equation*}
Now we can apply \cite[Thm.~3.1]{Stab} to the system \eqref{Df} and assert that when $\lambda_1(D_f) \leq 0$, i.e., $\tau \leq 1/\lambda_1(A)$, $P_0$ is a globally asymptotically stable equilibrium of \eqref{mat2}.
	
Conversely, if $\lambda_1(D_f)>0$, i.e. $\tau > 1/\lambda_1(A)$, there exists another equilibrium point $P_{\infty}$. $P_0$ and $P_{\infty}$ are the only equilibrium points in $I_N$ and $P_{\infty}$ is globally asymptotically stable in $I_N - \left\{0\right\}$.
	
Finally, since  $\tau >  1/\lambda_1(A)$, we have $\lambda_1(D_f)>0$. \rev{From Lyapunov's Linearization (or First) Method, it follows that} $P_0$ is an unstable equilibrium point in $I^N$.  
\end{proof}

\section{\so{Equitable Partitions}}\label{sec:equi}

\so{In this section we describe the SIS individual-based model for graphs with equitable partitions.}
The original definition of equitable partition is due to Schwenk \cite{Schwenk}.

\begin{definition}\label{def:eqpart}
Let $G=(V,E)$ be a graph. The partition $\pi=\left\{V_1,...,V_n\right\}$ of the node set $V$ is called \emph{equitable} if  
\so{ for all $i,j \in \left\{1, \dots ,n \right\}$}, there is an integer $d_{ij}$ such that 
\begin{equation*}
d_{ij}=\mbox{\rm deg}(v,V_j):=\# \left\{e \in E : e=\left\{v,w\right\}, w \in V_j \right\}.
\end{equation*} 
independently of $v \in V_i$.
\end{definition}

We shall identify the set of all nodes in $V_i$ \rev{with} the $i$-th {\em community} of the whole population.
In particular, each $V_i$ induces a subgraph of $G$ that is necessarily regular.

%
\begin{remark}\label{rem1}
We use the notation $\lcm$ and $\gcd$ to denote the least common multiple and greatest common divisor, respectively.
We can observe that the partition of a graph is equitable if and only if  
\begin{equation}
d_{ij} = \alpha \frac{\lcm(k_i,k_j)}{k_i}\nonumber
\end{equation}
where $\alpha$ is an integer satisfying $1 \leq \alpha \leq \mathop{\gcd}(k_i,k_j)$ and $k_i$ the number of nodes in $V_i$, for all $ i=1,...,n$.
\end{remark}

An equitable partition generates the \emph{quotient graph} $G/\pi$, which is a
\emph{multigraph} with cells as vertices and $d_{ij}$ edges between $V_i$ and $V_j$. 
For the sake of explanation, in the following we will identify $G/\pi$ \rev{with the} 
(simple) 
graph having the same cells vertex set, and where an edge exists between 
$V_i$ and $V_j$ if at least one exists in the original multigraph. We shall denote by $B$ the adjacency matrix of the graph $G/\pi$.
\begin{remark}\label{rem2}
In 
\cite{Bonaccorsi} 
 \rev{it has been considered the special case} when each community has a clique structure, i.e, $d_{ii}=k_i-1$ for all $i=1,...,n$. Moreover all nodes belonging to two linked 
communities $i$ and $j$ are connected, $d_{ij}=k_j$.  By considering the theory of equitable partition, 
we generalize the cited work and consider any kind of regular graph to represent the 
internal structure of each community. Moreover, unlike before, if two communities $i$ and $j$ 
are connected, each node in community $i$ is connected with \rev{$d_{ij}\leq k_j$ nodes in community $j$.}
\end{remark}

\subsection{Example} \label{subset:example}

\rev{Let us assume that the adjacency matrix $B$ of the quotient graph is given and that,}
for any \so{$i,j \in \left\{1, \dots, n\right\}$}, $b_{ij} \not=0$ implies $d_{ij} = k_j$, i.e., each node in $V_i$ is connected with every node inside $V_j$.
We can explicitly write the adjacency matrix $A$ in a block form. Let
%
$C_{V_{i}}=(c_{ij})_{k_i \times k_i}$ be the adjacency matrix of the 
subgraph induced by $V_i$
and $J_{k_i \times k_j}$ is an all ones $k_i \times k_j$ matrix; then
\begin{equation}\label{e:**}
A=
\begin{bmatrix} 
C_{V_1}&\varepsilon J_{k_1\times k_2}b_{12}&.&.& \varepsilon J_{k_1 \times k_n}b_{1n} \\
\varepsilon J_{k_2\times k_1}b_{21}&C_{V_2}&.&.& \varepsilon J_{k_2 \times k_n}b_{2n}\\
.&.&.&.&.\\
.&.&.&.&.\\
.&.&.&.&C_{V_n}
\end{bmatrix}
\end{equation}

\rev{We observe that \eqref{e:**} represents a block-weighted version of the adjacency matrix $A$. 
The derivation of NIMFA for the case of two different infection rates, considered in this paper, results 
in the replacement of the unweighted adjacency matrix in the NIMFA system \eqref{mat2} with its weighted 
version (see \cite{scoglio} for a deeper explanation)}.

\subsection{The quotient matrix}\label{sec:Q}

\so{We search for a smaller matrix $Q$ that contains the relevant information for the evolution of the system.}
Such a matrix is the \emph{quotient matrix} of the equitable partition. 
\rev{ In Prop. \ref{cor} we will see that} $Q$ and $A$ have the same spectral radii. As a consequence, we can compute its spectral radius in order to estimate the epidemic threshold, instead of computing \rev{the spectral
radius of matrix $A$.}

The quotient matrix $Q$ can be defined for any equitable partition: in view of the internal structure of a graph with an equitable partition, it is natural to consider the cell-wise average value of a function on the node set, that is to say the projection of the node space into the subspace of cell-wise constant functions. 
\begin{definition}\label{def:proj}
Let $G=(V,E)$ a graph. Let $\pi = \{V_i,\ i = 1, \dots, n\}$ be any partition of the node set $V$, let us consider the $n \times N$ matrix $S=(s_{iv})$, 
where
\begin{equation*}
s_{iv}=\begin{cases}
\frac{1}{\sqrt{|V_i|}} & \text{ $v \in V_i$}\\
0 & \text{otherwise}.
\end{cases}
\end{equation*}
The \emph{quotient matrix} of $G$ (with respect to the given partition) is 
$$Q:=SAS^T.$$
\end{definition}
Observe that by definition $SS^T=I$. 
 
In the case of the example in Sec. \ref{subset:example} the form of $Q$ is rather simple:
\begin{equation*}\label{eq:qii}
q_{ii}=
\sum_{h=1}^{k_i}
\left(\frac{1}{\sqrt{k_i}}\right)^2\sum_{k=1}^{k_i}(C_{V_i})_{kh}=\frac{1}{k_i}\sum_{h,k=1}^{k_i}(C_{V_i})_{kh}
\end{equation*}
and 
\begin{equation*}\label{eq:qij}
q_{ij}=
\frac{1}{\sqrt{k_i k_j}}\sum_{z \in V_i, \\ l \in V_j} a_{zl}= \sqrt{k_ik_j}\varepsilon b_{ij}.
\end{equation*}
Hence we obtain that 
\begin{equation*}\label{Qspec}
 Q=\diag(d_{ii})+ (\sqrt{k_ik_j}\varepsilon b_{ij})_{i,j=1,...n},
\end{equation*}
where $d_{ii}=\frac{1}{k_i}\sum_{h,k=1}^{k_i}(C_{V_i})_{kh}$ is the internal degree of the subgraph induced by $V_i$.

In the case of general equitable partitions,  
the expression for $Q$ writes
\begin{equation*}\label{eq:q}
 Q=\diag(d_{ii})+ (\sqrt{d_{ij}d_{ji}}\varepsilon b_{ij})_{i,j=1,...n}.
\end{equation*}

There exists a close relationship between the spectral properties of $Q$ and that of $A$. Being the order of $Q$ smaller of that of $A$, a result in
\cite{godsil} basically shows that $\sigma(Q)\subseteq\sigma(A)$, \rev{where with $\sigma(A)$ we refer, hereafter, to the spectrum of a square matrix $A$. Furthermore it holds the following} 
\begin{proposition}\label{cor}
Let $G=(V,E)$ a graph. Let $\pi = \{V_i,\ i = 1, \dots, n\}$ be an equitable partition of the node set $V$.
The adjacency matrix $A$ and the quotient matrix $Q$  have the same spectral radius, i.e.
\[
\lambda_1(Q)=\lambda_1(A).
\]
\end{proposition}

\begin{proof} \rev{See \cite[art. 62]{Graph}}.
 %
%
\end{proof}

\subsection{Complexity reduction}

\fdp{ Prop.~\ref{cor} further details that, once the network structure is encoded in the connectivity of a quotient graph $Q$, then the epidemic threshold $\tau_c^{(1)}$ is expressed by the spectral radius of $Q$.}

 \fdp{Now, since the order of $Q$ is smaller than the order of $A$, this can provide a computational advantage. The complexity reduction can be evaluated easily, e.g, in the case of the power iteration method \cite{MatAn}. The power iteration method is a numerical technique for approximating a dominant eigenpair of a diagonalizable matrix $L$, 
 using the following iteration 
\begin{equation*}
y^{h}=L \, y^{h-1}, \quad h=1,2,\ldots 
\end{equation*}
for a given initial vector $y^{0}$. As the iteration step $h$ increases, $y^{h}$ approaches a vector which 
is proportional to a dominant eigenvector of $L$. If we order the eigenvalues of $L$ such as as $|\lambda_1(L)|\geq |\lambda_2(L)| \geq \ldots \geq|\lambda_n(L)|$, the rate of convergence of the method is ruled by $ |\lambda_2(L)|/|\lambda_1(L)|$.}

 \fdp{In our case, for the Perron-Frobenius Theorem the dominant eigenvalue $\lambda_1(A)$ is positive and by Prop. \ref{cor}, $\lambda_1(A)=\lambda_1(Q)$. Furthermore $\sigma(Q) \subseteq \sigma(A)$, hence  $\max_{i \geq 2} |\lambda_i (A)| \geq \max_{i \geq 2} |\lambda_i (Q)|$: this means that 
the convergence of power iteration for matrix $Q$ is never slower than for matrix $A$. Finally, it is immediate that at each step the computational complexity is $O(n^2)$ for $Q$ whereas for $A$ it is  $O(N^2)$.
}
\subsection{A lower bound for $\tau^{(1)}_c$}
We can write $Q=D + \hB$, where $D=\diag(d_{ii})$ and $\hB=(\sqrt{d_{ij}d_{ji}}\varepsilon b_{ij})_{i,j=1,...n}$. %
\so{By the Weyl's theorem \cite{MatAn} we have}
\begin{equation}\label{lowbound}
\lambda_1(Q) \leq \lambda_1(D)+\lambda_1(\hB)=\max_{1\le i\le n} d_{ii} + \lambda_1(\hB).
\end{equation} 
\rev{From  (\ref{tauc}) and by Proposition \ref{cor} }
\begin{equation*}
\tau^{(1)}_c = 1/\lambda_1(A)=1/\lambda_1(Q),  
\end{equation*}
\rev{thus a lower bound for the epidemic threshold can be derived from (\ref{lowbound})}
\begin{equation}\label{low2}
\tau^{(1)}_c \geq  \tau^\star  
= \min_{i} \frac{1}{d_{ii} + \lambda_1(\hB)} ,
\end{equation}

In applications, when designing or controlling a network, this value can be adopted to determine a safety region 
$\{\tau \le \tau^\star\}$ for the effective spreading rate that guarantees the extinction of epidemics.

 \fdp{Fig.~\ref{fig:low} reports on the comparison of the lower bound and the actual threshold \rev{value}: \rev{it refers to} the case of a sample equitable partition composed of interconnected rings for increasing values of the community order.}

\rev{We observe that obtaining a lower bound for $\tau_c^{(1)}$ is meaningful because $\tau_c^{(1)}$ is itself a lower bound for the epidemic threshold $\tau_c$ of the exact stochastic model, i.e. $\tau_c = \alpha \tau_c^{(1)}$ with $\alpha \geq 1$,  as anticipated in Sec.~\ref{Epidemic Threshold}}. In fact, smaller values of the effective spreading rate $\tau$, namely $\delta>\beta/\tau_c^{(1)}$, correspond, in the exact stochastic model, to a region where the decay towards the healthy state decreases exponentially fast in time. %
\rev{By forcing the effective spreading rate below $\tau^*$, one ensures that the epidemic will go extinct in a reasonable time frame (we recall that, above the threshold, the overall-healthy state is only reached after an unrealistically long time.).}
\begin{figure}[t]	
\centering
\includegraphics[width=0.6\textwidth]{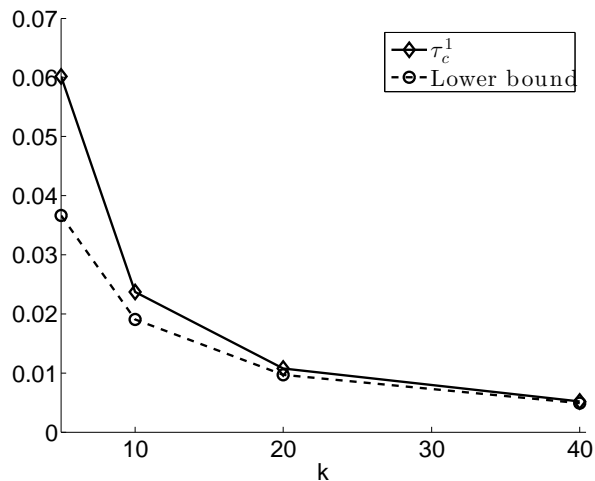}
\caption{\so{Lower bound \eqref{low2} versus epidemic threshold: comparison for different values of $k$ in a $40$-communities network. The internal structure of each community is a ring and $d_{ij}=2$ for all $i,j=1, \ldots, n$.}}\label{fig:low}
\end{figure}

\so{Equality can be attained in \eqref{low2}: consider for instance the graph described by the adjacency matrix $A$ in \eqref{e:**}. Furthermore, we may require that all $V_i$'s  have the same number of nodes $k_i=k$ and same internal degree $d_{ii}=d$, $i=1,\ldots,n$. In this case $Q= d\,{\rm Id}_n + \hB$, where $\hB:=(k \varepsilon b_{ij})_{i,j=1,...n}$, and
\begin{equation*}
\lambda_1(Q)= d + k\varepsilon \lambda_1(B),
\end{equation*} 
which is the exact value of $\lambda_1(A)$ and consequently of $\tau^{(1)}_c$.}

\section{Infection Dynamics for Equitable Partitions}\label{sec:InfDyn}

\so{In this section we show under which conditions matrix $Q$ can be used  in order to express the epidemic dynamics 
introduced in \eqref{mat2}. This allows us to describe the time-change of the infection probabilities by a system of $n$ 
differential equations instead of $N$. }

\so{
\begin{theorem}\label{reduction}
Let $G=(V,E)$ a graph and $\pi = \{V_j,\ j = 1, \dots, n\}$  an equitable partition of the node set $V$. Let $G_j$ be the subgraph of $G=(V,E)$ induced by cell $V_j$. 
If $p_h(0)=p_w(0)$ for all $h, w \in G_j$ and \rev{for all $j=1, \dots, n$}, then $p_h(t)=p_w(t)$  for all $t > 0$. In this case we can reduce the number of equations representing the time-change of infection probabilities using the quotient matrix $Q$. 
\end{theorem}
}

\begin{proof}
Let $\overline{p}_j(t)= \frac{1}{k_j}\sum_{h \in G_j} p_h(t)$ be
 the average value of the infection probabilities at time $t$ of nodes in $G_j$. 
Then starting from \eqref{mat2}, we can write a new system of differential equations
\begin{eqnarray}\label{mean}
 &&\frac{d \left(p_h(t) - \overline{p}_j(t)\right)}{dt}= - \delta (p_h(t)- \overline{p}_j(t))+ \beta (1-p_h(t)) \sum_{z=1}^N a_{h z} p_z (t)\nonumber\\ 
&&\hskip32mm - \frac{1}{k_j} \beta \sum_{l \in G_j} (1-p_l(t))\sum_{z=1}^N a_{lz}p_z(t), \qquad  \forall h \in G_j,  \quad j=1, \ldots, n. 
\end{eqnarray}
\rev{From \eqref{mean} we have
\begin{align*} 
\frac{d \left(p_h(t) - \overline{p}_j(t)\right)}{dt} &= - \delta (p_h(t)-\overline{p}_j(t)) + \beta \left(\sum_{m=1}^n \sum_{z \in G_m} a_{hz} p_z(t)- \frac{1}{k_j} \sum_{l \in G_j}\sum_{m=1}^n \sum_{z \in G_m} a_{lz} p_z(t)\right)\\ 
&  - \beta \left(p_h(t) \sum_{m=1}^n \sum_{z \in G_m} a_{hz} p_z(t) - \frac{1}{k_j} \sum_{l \in G_j} p_l(t) \sum_{m=1}^n \sum_{z \in G_m} a_{lz} p_z(t) \right),
\end{align*}
that can be written as
\begin{align*} 
\frac{d \left(p_h(t) - \overline{p}_j(t)\right)}{dt} &=  - \delta (p_h(t)-\overline{p}_j(t)) + \beta \left(\frac{1}{k_j}  \sum_{l \in G_j}  \sum_{m=1}^n \sum_{z \in G_m} \left(a_{hz}-a_{lz}\right)p_z(t)\right)\\
                & -\beta \frac{1}{k_j} \sum_{l \in G_j} \sum_{m=1}^n \sum_{z \in G_m}  \left( a_{hz} p_h(t)- a_{lz} p_l(t)\right) (p_z(t)  -\ovp_m(t))\\
								& - \beta\frac{1}{k_j} \left (  \sum_{l \in G_j} \sum_{m=1}^n \sum_{z \in G_m} \left(a_{hz}p_h(t)- a_{lz} p_l(t) \right)\right) \ovp_m(t).
\end{align*}
Whence, since  $\sum_{z \in G_m} a_{hz} = d_{jm}$, for $h \in G_j$ and for all $m=1, \dots, n$, we have
\begin{align}\label{mean11} 
\frac{d \left(p_h(t) - \overline{p}_j(t)\right)}{dt} &=  - \left[ \sum_{m=1}^n \beta d_{jm} \overline{p}_m(t)+ \delta \right](p_h(t)- \overline{p}_j(t)) \\
                      & + \beta \frac{1}{k_j}  \sum_{l \in G_j}  \sum_{m=1}^n \sum_{z \in G_m} \left(a_{hz}-a_{lz}\right)p_z(t)\nonumber \\ 
											& -  \beta \frac{1}{k_j} \sum_{l \in G_j} \sum_{m=1}^n \sum_{z \in G_m}  \left( a_{hz} p_h(t)- a_{lz} p_l(t)\right) (p_z(t)  -\ovp_m(t)). \nonumber
\end{align}
Now, we note that
$$
-   \frac{1}{k_j} \sum_{l \in G_j} \sum_{m=1}^n \sum_{z \in G_m}  \left( a_{hz} p_h(t)- a_{lz} p_l(t)\right) (p_z(t)  -\ovp_m(t))
$$
can be written as
$$
 - \frac{1}{k_j}  \sum_{l \in G_j}  \sum_{m=1}^n \sum_{z \in G_m} \left((p_h(t)-\ovp_j(t))a_{hz} - (p_l(t) - \ovp_j(t))a_{lz}\right)\left(p_z(t) - \ovp_m(t)\right)
$$
$$
-\frac{1}{k_j}  \sum_{l \in G_j} \sum_{m=1}^n \sum_{z \in G_m} \ovp_j(t) (a_{hz}-a_{lz})\left(p_z(t)-\ovp_m(t)\right),
$$
whence we can rewrite \eqref{mean11} as
\begin{align*} 
\frac{d \left(p_h(t) - \overline{p}_j(t)\right)}{dt} &=  - \left[ \sum_{m=1}^n \beta d_{jm} \overline{p}_m(t)+ \delta \right](p_h(t)- \overline{p}_j(t)) \\
   & + \beta \frac{1}{k_j}  \sum_{l \in G_j}  \sum_{m=1}^n \sum_{z \in G_m} \left(a_{hz}-a_{lz}\right) \left(p_z(t)- \ovp_m(t) + \ovp_m(t) \right)\\
	& - \beta \frac{1}{k_j}  \sum_{l \in G_j}  \sum_{m=1}^n \sum_{z \in G_m} \left((p_h(t)-\ovp_j(t))a_{hz} - (p_l(t) - \ovp_j(t))a_{lz}\right)\left(p_z(t) - \ovp_m(t)\right)\\
	&  - \beta \frac{1}{k_j}  \sum_{l \in G_j} \sum_{m=1}^n \sum_{z \in G_m} \ovp_j(t) (a_{hz}-a_{lz})\left(p_z(t)-\ovp_m(t)\right).
\end{align*}
Finally, since $\frac{1}{k_j}  \sum_{l \in G_j}  \sum_{m=1}^n \sum_{z \in G_m} \left(a_{hz}-a_{lz}\right)\ovp_m(t) = 0$,  
we can consider the following system 
\begin{eqnarray}\label{mean2}
&&\frac{d \left(p_h(t) - \overline{p}_j(t)\right)}{dt}= -\left[ \sum_{m=1}^n \beta d_{jm} \overline{p}_m(t)- \delta \right](p_h(t)- \overline{p}_j(t)) \nonumber\\
&&+ \beta \frac{1}{k_j} \sum_{l \in G_j}\sum_{m=1}^n \sum_{z\in G_m} (a_{hz}-a_{lz})(p_z(t) -\overline{p}_m(t))(1-\overline{p}_j(t))
\nonumber\\
&&-\beta \frac{1}{k_j} \sum_{l \in G_j} \sum_{m=1}^n \sum_{z\in G_m} ((p_h(t)- \overline{p}_j(t))a_{h z}- (p_l(t)-\overline{p}_j(t))a_{lz})(p_z(t)-\overline{p}_m(t)),\nonumber\\ 
&&\hskip80mm  \forall h \in G_j,  \quad j=1, \ldots, n\nonumber
\end{eqnarray}}

Now let us denote by $g(t)$ the solution of \eqref{mean}, where $g: \mathbb{R} \rightarrow \mathbb{R}^N$ and consider the case where 
\begin{equation}\label{SameIniz}
p_h(0)- \overline{p}_j(0)=0,   \quad \forall h \in G_j, \quad j=1, \ldots, n,
\end{equation}
i.e., $p_h(0)=p_w(0)$ for all $h, w \in G_j$. Then, from \eqref{mean2}, we can easily see that the identically zero function $g \equiv 0$ 
is the unique solution of \eqref{mean} with initial conditions \eqref{SameIniz}.
 Indeed $g\equiv 0$ means that for all $t \geq 0$,
 $p_h(t)=p_w(t)$ for all $h,w \in G_j$, $j=1, \ldots, n$. \rev{Moreover the vector $P(t)$ such that $p_h(t)=p_w(t)$ for all $h,w \in G_j$, $j=1, \ldots, n$, is a solution of \eqref{mat2} and it is unique in $[0,1]^N$ with respect to the initial conditions \eqref{SameIniz}, \cite[Cap. 2, Sec. 2.2]{DiffEqandDynamicalSystem}. Thus we can conclude that also $g=0$ is a unique solution of \eqref{mean} in $[-1,1]^N$.}

Basically we have shown that the following subset of $I_N$
\begin{eqnarray*}
M= \left\{ P \in [0,1]^N | p_1= \ldots = p_{k_1}= \overline{p}_1,  p_{k_1+1}= \ldots = p_{k_1 + k_2}= \overline{p_2} ,\right . \nonumber\\
 \left . \ldots,  p_{(k_1+..k_{n-1}+1)}= \ldots = p_{N} =\overline{p}_n\right\}
\end{eqnarray*}
is a positively invariant set for the system \eqref{mat2}.
This allows us to reduce the system \eqref{mat2} of $N$ differential equations and describe the time-change of the infection probabilities by a system of $n$ equations involving the matrix $Q$.

 Indeed, let us consider  $P(0) \in M$ and $\overline{P}=(\ovp_1, \dots, \ovp_n) $, we can write
\begin{eqnarray}\label{eq:red_sys_1}
\frac{d\overline{p}_j(t)}{dt}& =& \beta (1-\overline{p}_j(t))\sum_{m=1}^n \varepsilon b_{jm} d_{jm} \overline{p}_m(t)  \\
			            & +& \beta d_j(1-\overline{p}_j(t))\overline{p}_j(t)- \delta \overline{p}_j(t),  \qquad j=1,\ldots,n\nonumber
\end{eqnarray}
Hence, based on Thm. 2.1 in \cite{godsil}, we observe that
\begin{equation*}
q_{ij}=(k_j/k_i)^{1/2}d_{ji},
\end{equation*}
This relation in our case brings
\begin{equation*}
 d_{jm}=\left(\frac{k_j}{k_m}\right)^{- 1/2}\frac{ q_{mj}}{\eps}=\left(\frac{k_j}{k_m}\right)^{- 1/2}\frac{q_{jm}}{\eps},
\end{equation*}
 where the last equality holds because $Q$ is symmetric. We can rewrite (\ref{eq:red_sys_1}) as 
%
\begin{eqnarray}\label{eq:red_sys_2}
\frac{d\overline{p}_j(t)}{dt}&& = \beta(1-\overline{p}_j(t))\sum_{m=1, m \neq j}^n \left(\frac{k_j}{k_m}\right)^{- 1/2}%
q_{jm} p_m(t) \nonumber \\
                  && +\beta q_{jj}(1-\overline{p}_j( t))p_j(t)- \delta \overline{p}_j(t) ; \qquad j=1,\ldots,n
\end{eqnarray}
where $q_{jj}=d_{jj}= \lambda_1(C_{V_j})$.
The matrix representation of  \eqref{eq:red_sys_2}  
is the following
\begin{equation}\label{eq:red_sys_3}
\frac{d\overline{P}(t)}{dt}= \beta \left({\rm I}_n- \operatorname{diag} (\overline{p}_j(t)) \right) \widetilde Q \overline{P}(t) - \delta \overline{P}(t),
\end{equation}
where $\widetilde Q= \operatorname{diag}\left(\frac{1}{\sqrt{k_j}}\right) Q\operatorname{diag} (\sqrt{k_j})$. It is immediate to observe that $\sigma(Q)= \sigma(\tilde{Q})$.%
%
\end{proof}

\begin{corollary}\label{corSteady}
When $\tau > \tau^{(1)}_c$ the non-zero steady-state $P_\infty$ of the system \eqref{mat2} belongs to $M -\left\{0\right\}$.
\end{corollary}
\so{
\begin{proof}
In Theorem \ref{thresh} we have shown that when $\tau >\tau^{(1)}_c$, the system \eqref{mat2} has a  globally asymptotically stable equilibrium $P_\infty$ in $I_N - \left\{0\right\}$; hence for any initial state $P(0) \in I_N - \left\{0\right\}$ 
\begin{equation*}
\lim_{t \rightarrow \infty} ||P(t) - P_{\infty}||= 0.
\end{equation*}
We have  proved in Thm. \ref{reduction} that if $P(0) \in M$ then $P(t) \in M$ for all $t> 0$, thus we can conclude that $P_\infty$ must be in $M- \left\{0\right\}$ when $\tau >\tau^{(1)}_c$. 
\end{proof}}

\so{Basically, Corollary~\ref{corSteady} says that one can compute the $n \times 1$ vector, $\overline{P}_{\infty}$, of the reduced system  \eqref{eq:red_sys_3} in order to obtain the $N \times 1$ vector, $P_{\infty}$, of \eqref{mat2}: indeed   $ p_{z \infty}, \ldots , p_{x \infty}=\overline{p}_{j \infty}$,  for all $z, x \in G_j$ and $j=1, \ldots, n$. This provides a computational advantage by solving a system of $n$ equations instead of $N$.}
Moreover, since $P_{\infty}$ is a globally asymptotically stable equilibrium in $I^N-\left\{0\right\}$, the trajectories starting outside $M$ will approach those starting in $M-\left\{0\right\}$. The same holds clearly for trajectories starting in $I^N$ and in $M$ when $\tau \leq \tau^{(1)}_c$. Numerical experiments in Fig.~\ref{fig:averaged} depict this fact.

\so{The statements proved above can be easily \rev{verified}, with a direct computation, in the simple case of graphs considered in \cite{Bonaccorsi} (see Remark \ref{rem2}). Indeed for all $h,w \in G_j$, $j=1, \ldots, n$, we have
\begin{eqnarray}\label{compexp}
\frac{d(p_h(t)-p_w(t))}{dt}=&& -\, \delta \, (p_h(t)-p_w(t)) + \beta \sum_{z \notin G_j} \left[(1-p_h(t)) a_{hz} - (1-p_w(t))a_{wz}\right] p_z(t)\nonumber\\
&&+ \, \beta \ \sum_{z \in G_j,  z \neq h,w} \left[(1-p_h(t)) a_{hz} - (1-p_w(t))a_{wz}\right] p_z(t) \nonumber\\
&&+\, \beta \sum_{z= h,w} \left[(1-p_h(t)) a_{hz} - (1-p_w(t))a_{wz}\right] p_z(t)
\end{eqnarray}
Since in this special case $a_{hz}=a_{wz}$, for all $z \in V$ s.t.  $z \neq h,j$, we can rewrite \eqref{compexp} as
\begin{equation*}
\frac{d(p_h(t)-p_w(t))}{dt}= - \left[\delta + \beta\left( \sum_{z=1, z \neq h,w} ^N a_{hz} p_z(t) + 1 \right) \right] \left(p_h(t)-p_w(t)\right).
\end{equation*}
whence
\begin{equation*}\label{sol}
p_h(t)-p_w(t)= \left(p_h(0)-p_w(0)\right) e^{-\int_0^{t} \delta + \beta\left( \sum_{z=1, z \neq h,w} ^N a_{hz} p_z(s) + 1\right) ds }.
\end{equation*}
Thus if $p_h(0)=p_w(0)$ for the uniqueness of solution it will occur $p_h(t)=p_w(t)$ for all $t>0$, as we have proved in Thm. \ref{reduction}, but if the initial conditions are different, the distance between $p_w(t)$ and $p_z(t)$ decreases exponentially.}

\begin{remark}
\fdp{The framework of quotient graphs extends the NIMFA model to graphs with prescribed community network structure. It reduces to the original NIMFA model when $k_j=1$ for all $j=1,..,n$.}
\end{remark}

\subsection{\so{Steady-state}}

We focus now on the \so{computation of the} steady-state \rev{$P_\infty = \big(p_{i\infty} \big)_{i=1,\dots,N}$} of system \eqref{mat2}.
\so{To this aim, by Corollary \ref{corSteady}, we can compute the steady-state \rev{$\overline{P}_\infty = \big(\ovp_{j\infty} \big)_{j=1,\dots,n}$} of the reduced system \eqref{eq:red_sys_3}  and obtain}
\begin{equation*}
 \beta(1-\ovp_{j \infty})\sum_{m=1}^n \left(\frac{k_j}{k_m}\right)^{- 1/2}
q_{jm} \ovp_{m\infty}- \delta \ovp_{j \infty}=0, \qquad j=1, \dots, n
\end{equation*}
\so{whence
\begin{eqnarray}\label{meta}
\ovp_{j \infty} &=& \frac{\beta \sum_{m=1}^n \left(\frac{k_j}{k_m}\right)^{- 1/2}
q_{jm} \ovp_{m \infty}}{\beta \sum_{m=1}^n \left(\frac{k_j}{k_m}\right)^{- 1/2}
q_{jm} \ovp_{m \infty}+ \delta}\nonumber = 1-\frac{1}{1+ \tau \sum_{m=1}^n \left(\frac{k_j}{k_m}\right)^{- 1/2}
                  q_{jm} \ovp_{m \infty}}\nonumber \\
   \hskip -2mm         &=& 1-\frac{1}{1+\tau g_j\left(\overline{P}\right)}
\end{eqnarray}}
\so{where
{\small\begin{equation*}
g_j\left(\overline{P}\right):=\left(d_{jj}+ \eps\sum_{m=1}^n \left(\frac{k_j}{k_m}\right)^{- 1/2} \sqrt{d_{jm}d_{mj}}\right) 
   -\sum_{m=1}^n \left(\frac{k_j}{k_m}\right)^{- 1/2}\!\!q_{jm}(1-\ovp_{m\infty}).
\end{equation*}}}

\so{From \eqref{meta} follows that the steady-state infection probability of any node $j$ is bounded by
\begin{equation}\label{boundmeta}
0 \leq \ovp_{j \infty} \leq 1-\frac{1}{1+\tau(d_{jj}+ \eps\sum_{m=1}^n \left(\frac{k_j}{k_m}\right)^{- 1/2} \sqrt{d_{jm}d_{mj}})},
\end{equation}
where the inequality holds true because $\ovp_{j \infty} \in [0,1]$ for all $j=1, \dots, n$.}

\so{By introducing $1-\ovp_{m \infty}=\frac{1}{1+\tau \sum_{z=1}^n \left(\frac{k_m}{k_z}\right)^{- 1/2} q_{mz} \ovp_{z \infty}}$ in \eqref{meta}, we can express $\ovp_{j \infty}$ as a continued fraction iterating the formula
\begin{equation*}
x_{j,s+1}=f(x_{1;s},..,x_{n;s})
=1- \frac{1}{1+ \tau g_j(x_{1;s},..,x_{n;s})}, 
\end{equation*}}
\so{As showed in~\cite{VanMieghem2009}, after a few iterations of the formula above, one can obtain a good approximation of $\ovp_{j \infty}$, with a loss in the accuracy of the calculation around $\tau=\tau_c$. Ultimately, such numerical estimation can be used to improve the bound in (\ref{boundmeta}).
}

\so{If we consider a regular graph where communities have the same number of nodes, then  
\begin{equation*}\label{eq:approx}
\ovp_{j \infty}= 1-\left(1/\tau \left(d_{jj}+ \eps\sum_{m=1}^n \left(\frac{k_j}{k_m}\right)^{- 1/2} \sqrt{d_{jm}d_{mj}}\right)\right)
\end{equation*}
is the exact solution of \eqref{meta}.}

Now let  
$ r_j= d_{jj}+ \eps\sum_{m=1}^n \left(\frac{k_j}{k_m}\right)^{- 1/2} \sqrt{d_{jm}d_{mj}}$ and $r(1)=\min_j r_j$;
relying on the estimate $\ovp_{j \infty} \approx 1-\left(1/\tau r_j\right)$ we can express the steady-state average fraction  of infected nodes $y_{\infty}(\tau)=(1/N)\sum_{j=1}^n k_j p_{j \infty}(\tau)$ by

\begin{equation}\label{frac}
y_{\infty}(\tau) \approx 1 - \frac{1}{\tau N} \sum_{j=1}^n k_j \frac{1}{d_{jj}+ \eps\sum_{m=1}^n \left(\frac{k_j}{k_m}\right)^{- 1/2}\sqrt{d_{jm}d_{mj}}}.
\end{equation}

\rev{According to the analysis reported in \cite{VanMieghem2009}, approximation \eqref{frac} becomes the more precise the more the difference $r(2)-r(1)$ is small, where $r(2)$ is the second smallest of the $r_j$'s.
 In Sec.~\ref{exp} we report on some related numerical experiments.}

\subsection{Examples}\label{ex}

In Fig.~\ref{fig:fig11} we provide an example of a graph which has an equitable partition {with respect to } $V_1=\{v_1\}$, $V_2=\{v_2,v_3\}$, $V_3=\{v_3,v_4,v_5,v_6\}$, $V_4=\{v_7,v_8,v_9,v_{10},v_{11},v_{12},v_{13}\}\}$. 

The corresponding quotient matrix reads 
\[
Q=
\begin{bmatrix}
0&\varepsilon \sqrt{2} & \varepsilon 2 & 0 \\
\varepsilon \sqrt{2} & 1 &\varepsilon \sqrt{2} & \varepsilon \sqrt{3} \\
\varepsilon 2 & \varepsilon \sqrt{2} & 2 & 0\\
0&\varepsilon \sqrt{3} & 0 & 3 \\
\end{bmatrix}
\]
\so{From ~\eqref{eq:red_sys_3} we have that the steady-state can be computed by
\begin{equation*}\label{matM}
\overline{P}_{\infty}= \frac{\beta}{\delta}({\rm I}_n - \diag (\overline{p}_{\infty}))\diag(s_j) Q  \diag(1/s_j) \overline{P}_{\infty} , 
\end{equation*}
where \rev{$s_j$ is the $j$-th entry} of vector $s=(1, \sqrt{2}, 2, \sqrt{6})$.}

\subsection{Numerical experiments}\label{exp}
 In Figures~\ref{fig:fig2} and~\ref{fig:fig3} \so{we provide a comparison between the solution of 
the reduced ODE system \eqref{eq:red_sys_3} for the graph in Fig. \ref{fig:fig11} and the 
averaged $50\cdot 10^4$ sample paths resulting from a discrete event simulation} \rev{of the exact SIS process}.  \rev{The discrete event simulation is based on the generation} of independent Poisson processes for both the infection of healthy nodes and the recovery of infected ones. We observe that, as expected, NIMFA provides an upper bound to the dynamics of the infection probabilities. Also, in Fig.~\ref{fig:fig2} we observe that the dynamics for the communities that are initially healthy is characterized by a unique maximum for the  infection probability, which decreases afterwards. The communities initially infected, conversely, show a monotonic decrease of the infection probability.


\begin{figure}[th!]	
\centering
	\includegraphics[width=0.750\textwidth]{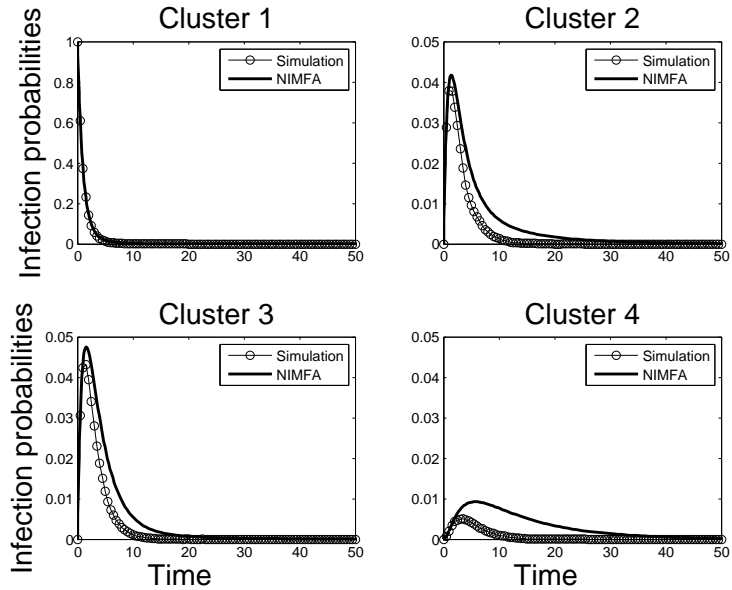}
	\caption{Dynamics of infection probabilities for each community of the network in Fig.\ref{fig:fig11}: simulation versus numerical solutions of (\ref{eq:red_sys_3}); $\tau = \beta/\delta < \tau_c^{(1)}=0.3178$, with $\beta=0.29$ and $\delta=1$, \rev{$\eps=0.3$}. At time $0$ the only infected node is node $1$.} 
	\label{fig:fig2}
\end{figure}

\begin{figure}[th!]	
	\centering
	\includegraphics[width=0.75\textwidth]{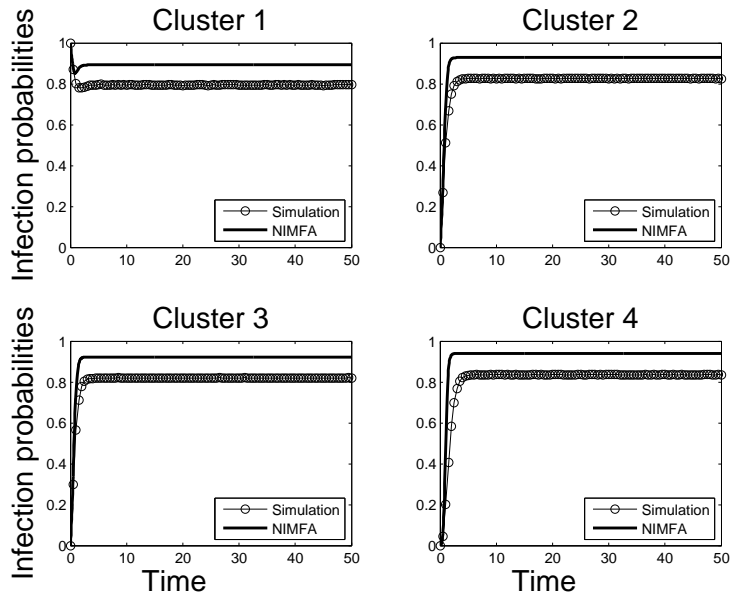}
	\caption{Dynamics of infection probabilities for each community of the network in Fig.\ref{fig:fig11}: simulation versus numerical solutions of (\ref{eq:red_sys_3}); $\tau = \beta/\delta > \tau_c^{(1)}=0.3178$, with $\beta=1.5$ and $\delta=0.3$, \rev{$\eps=0.3$}; initial conditions as in Fig.~\ref{fig:fig2}.} \label{fig:fig3}
\end{figure}

\rev{Fig.~\ref{fig:completo2} depicts the same comparison in the case of} a network with eighty nodes partitioned into four communities; each community is a complete graph and all nodes belonging to two linked communities are connected (see Remark~\ref{rem2}). The agreement between NIMFA and simulations improves compared to Fig.~\ref{fig:fig3}. This is expected, because the accuracy of NIMFA is known to increase with network order $N$, under the assumption that the nodes' degree also increases with the number of nodes. Conversely, it is less accurate, e.g., in lattice graphs or regular graphs with fixed degree not depending on $N$ ~\cite{VanMieghem2009,Accuracy}.

\begin{figure}[t]	
	\centering
	\includegraphics[width=0.80\textwidth]{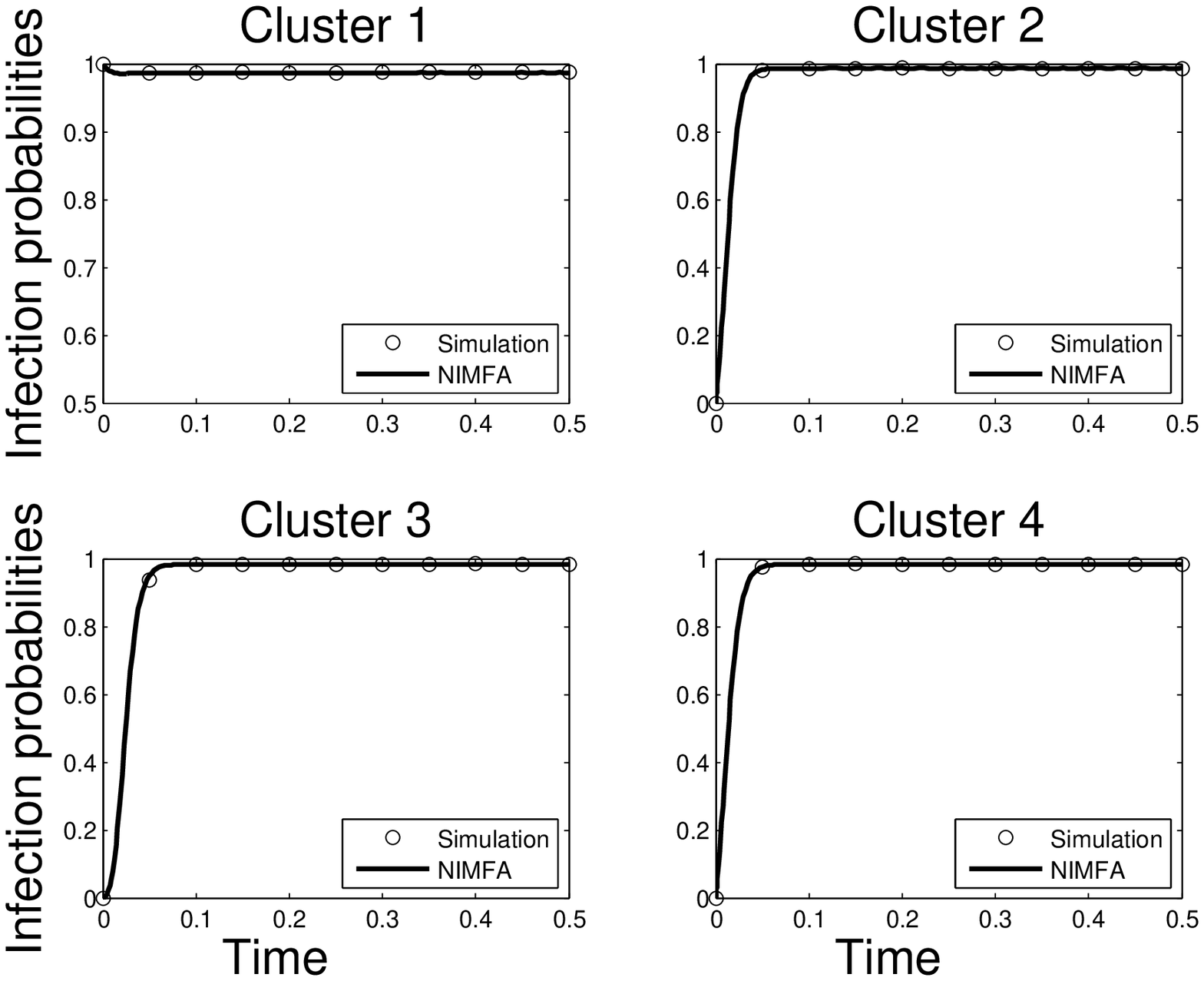}
	\caption{\so{Infection probabilities for each community in a network with $N=80$, $d_{ii}=k_i-1=19$ and $d_{ij}=20$, for all $i,j=1,..,4$: simulation versus numerical solutions of (\ref{eq:red_sys_3}); $\tau = \beta/\delta > \tau_c^{(1)}=0.0348$, with $\beta=5$ and $\delta=2$, \rev{$\eps=0.3$}; at time 0 all nodes of the 1-st community are infected.}}\label{fig:completo2}
\end{figure}
\rev{Fig.~\ref{fig:averaged} depicts the solutions of system~\eqref{mat2} for each node belonging to $V_3$ in the graph of Fig.~\ref{fig:fig11}; here nodes in $V_3$ have different initial infection probabilities $p_i(0)$'s. These solutions are compared with the one computed using the reduced system ~\eqref{eq:red_sys_3}, in the case when the initial conditions for those nodes are the same, precisely equal to the mean value of the $p_i(0)$'s. As expected, trajectories starting outside invariant set $M$ described in Thm.~\ref{reduction} tend to approach the one starting in $M$ as time elapses.}
 \begin{figure}[t]
 \centering
    \begin{minipage}{0.45\textwidth}
    \includegraphics[width=\textwidth]{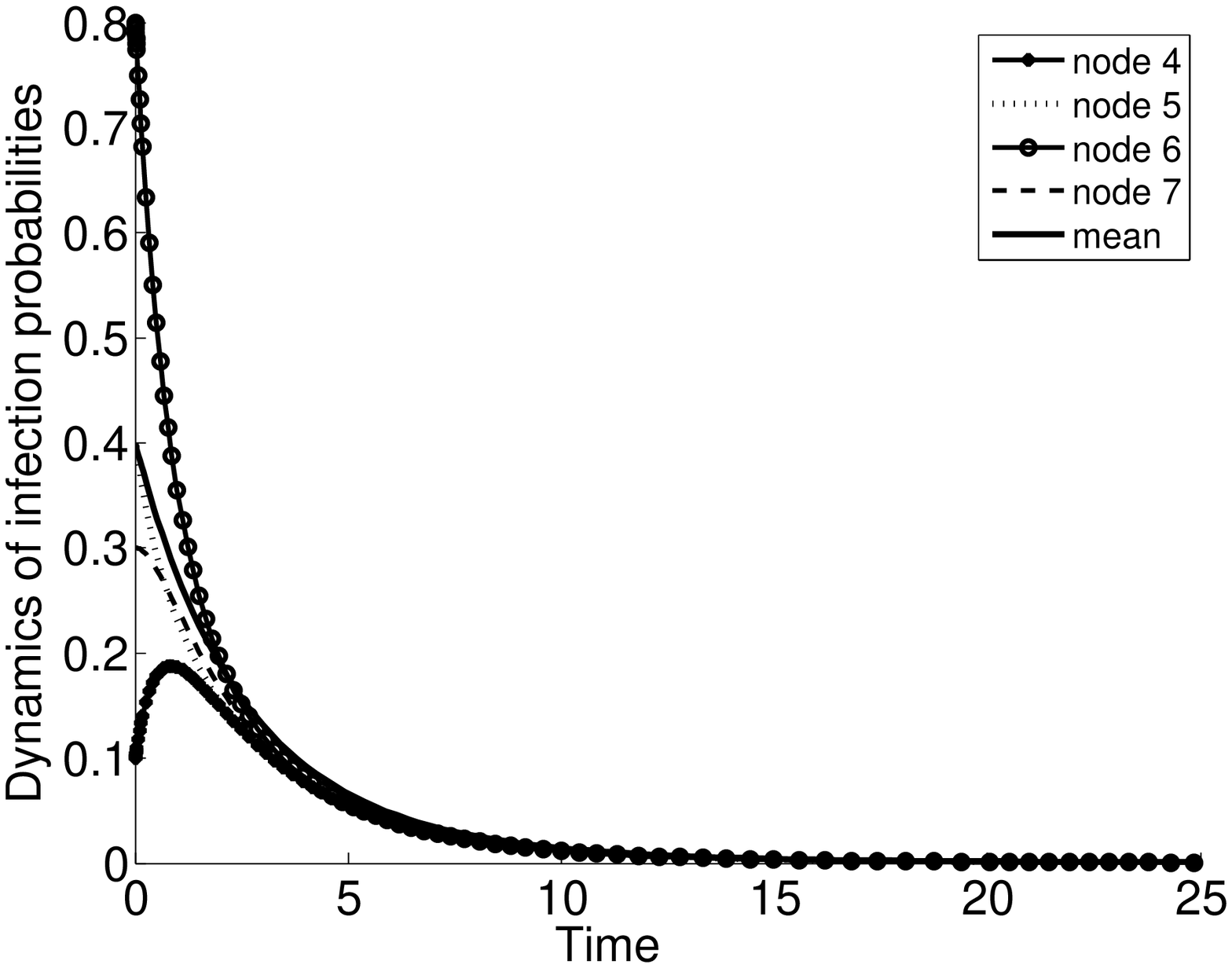}\put(-168,120){a)}
    \end{minipage}
    \hskip4mm
    \begin{minipage}{0.45\textwidth}
    \includegraphics[width=\textwidth]{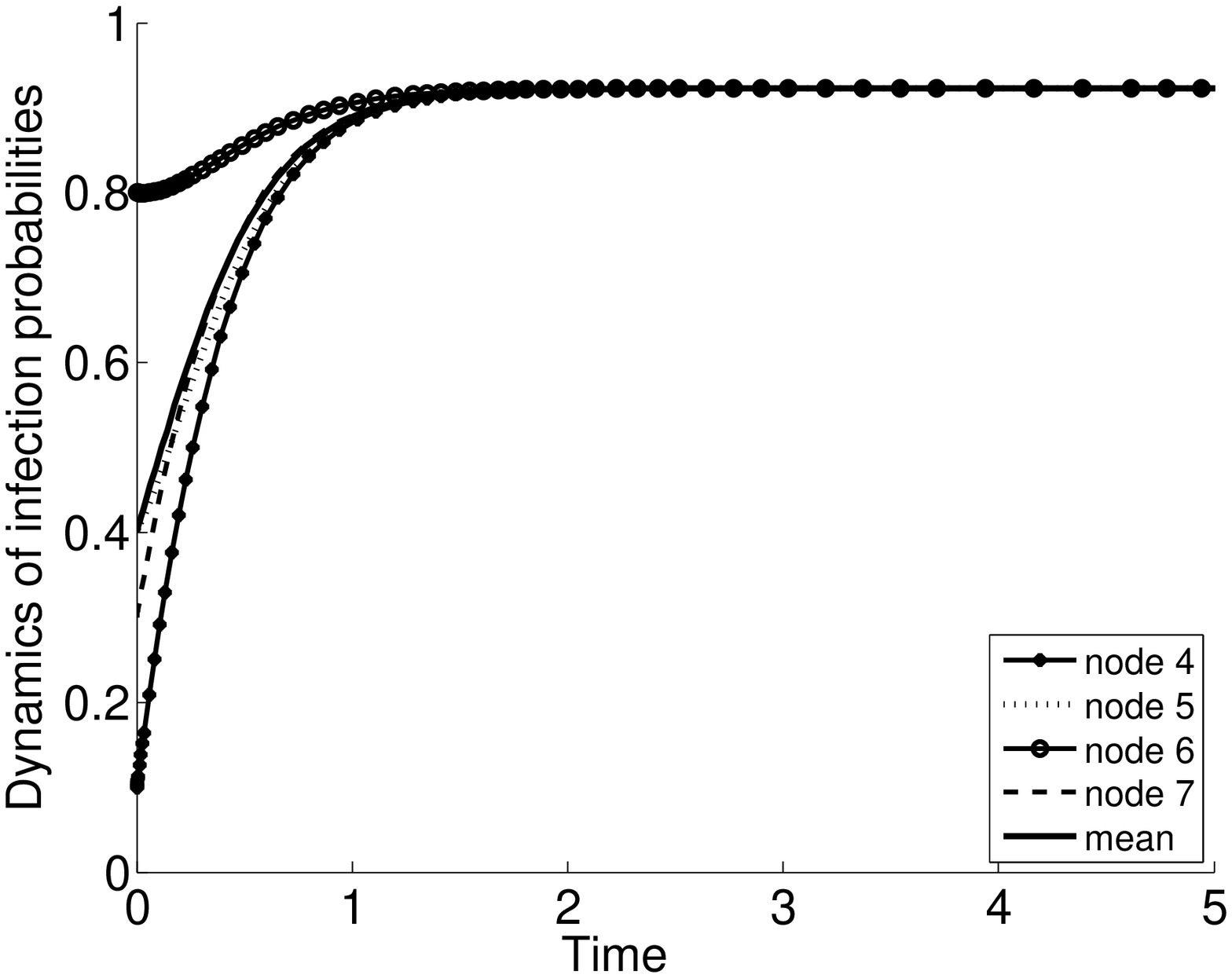}\put(-162,120){b)}
    \end{minipage}
    \caption{\so{Comparison between the dynamics of the original system~\eqref{mat2} for each of the nodes belonging to $V_3$ in Fig. \ref{fig:fig11}, for different initial conditions and the dynamics of the reduced system \eqref{eq:red_sys_3}. In the latter case the initial conditions for each node are the mean value of the $p_i(0)$s. a) case below the threshold: $\beta=0.29$, $\delta=1$, \rev{$\eps=0.3$} b) case above the threshold: $\beta=1.5$, $\delta=0.3$}, \rev{$\eps=0.3$}.}
    \label{fig:averaged}
  \end{figure}
\rev{Finally, we report on numerical experiments about the steady-state average fraction of infected nodes. More  precisely, Fig.~\ref{fig:frac1} compares the value obtained by solving the original system \eqref{eqpoints} and the value obtained from approximation \eqref{frac}, as a function of $\tau$.}

\so{In Fig. \ref{fig:SIS}, instead, we have reported on the comparison between the steady-state average fraction of infected nodes, as function of $\tau$, computed via NIMFA and via simulations. We consider a graph of regular degree $d=10$ and $N=500$, whose communities are clique, \rev{each with the same number of elements $k$. We repeat the same calculation for different values of $k$ in the communities}. As it can be observed, our model and the exact SIS model are in good agreement and the root mean square error between them decreases as $k$ increases.}

\begin{figure}[t]	
\centering
\includegraphics[width=0.60\textwidth]{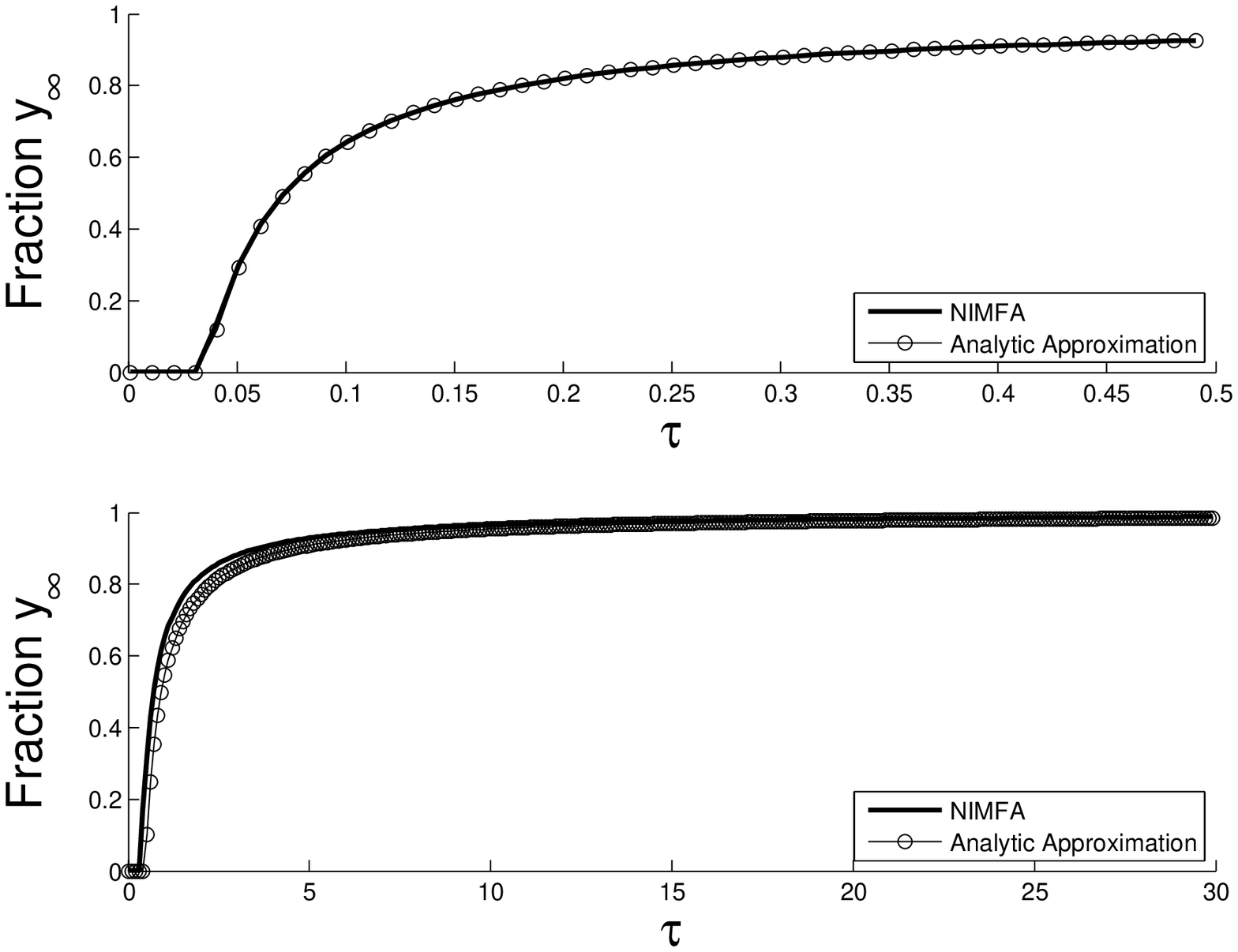}
\put(-220,160){a)}\put(-220,70){b)}
\caption{\so{Steady-state average fraction of infected nodes, for different values of $\tau$: comparison between the  approximation \eqref{frac} and the exact computation \eqref{eqpoints}; a) the graph is the one considered in Fig.~\ref{fig:fig11} and b) the one considered in Fig.~\ref{fig:completo2}.}}\label{fig:frac1}
\end{figure}
\begin{figure}[h]	
\centering
\includegraphics[height=50mm,width=0.60\textwidth]{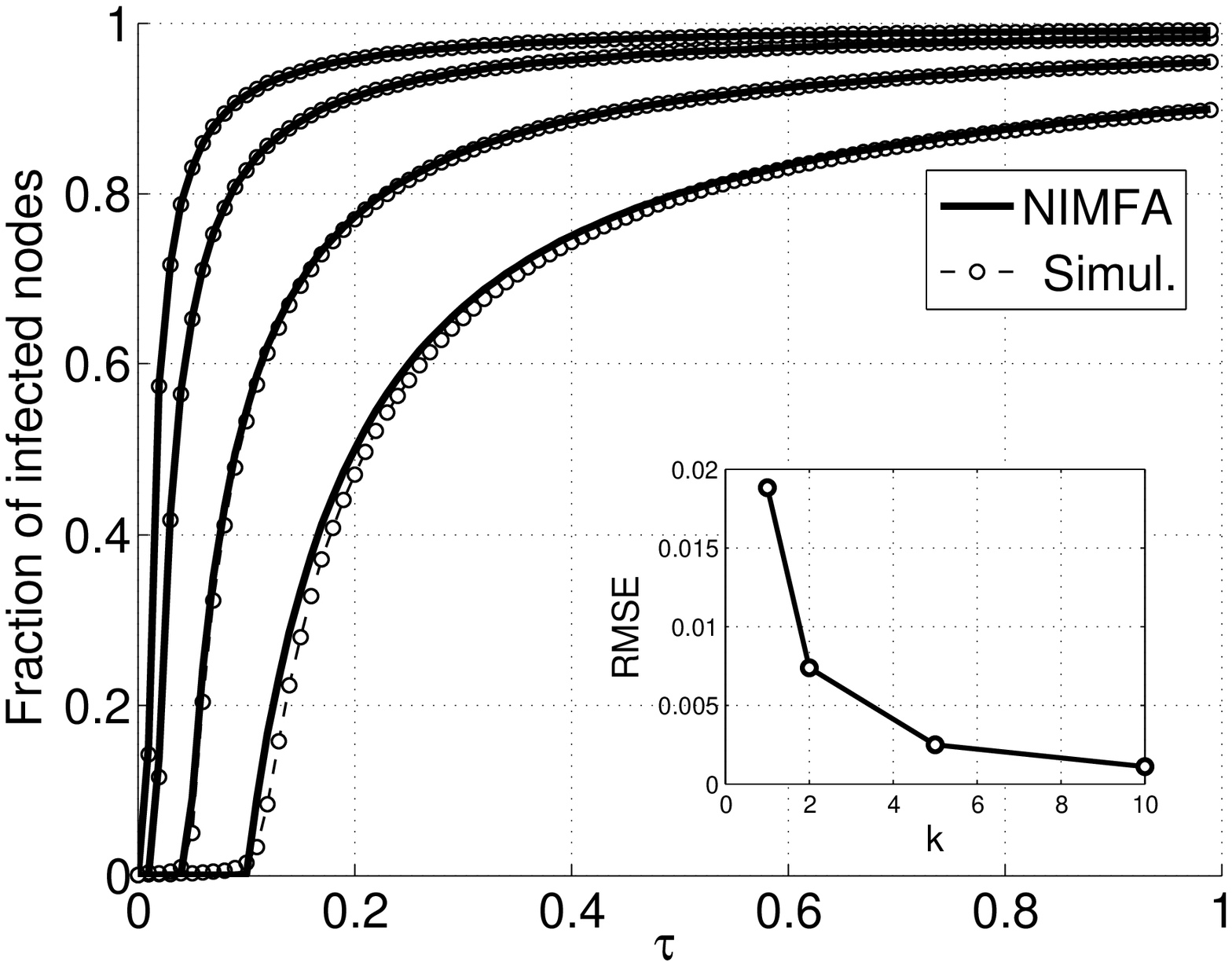}
		 \put(-120,80){\vector(-1,1){45}\hskip14mm $k=1,2,5,10$}
	\caption{\so{ Steady-state average fraction of infected nodes for different values of $k$ and $\tau$, 
	for a graph of regular degree $d=10$ and $N=500$. The internal structure of each community is a clique. Both NIMFA and simulations are shown. The inserted plot represents the root mean square error between the simulated and the approximated fraction of infected nodes.}}\label{fig:SIS}
  \end{figure}
\section{Almost equitable partitions}\label{AlmEq}

\rev{In this section we consider graphs where the partition of the vertex set is \textsl{almost equitable}. Thus, we can relax the initial assumption on the regularity of the internal community structure implied by the definition of equitable partition.}



\begin{definition}\label{de:aep}
The partition $\pi=\left\{V_1,...,V_n\right\}$ is called \emph{almost equitable} if \so{ for all $i,j \in \left\{1, \dots ,n \right\}$} with $i \neq j$, 
there is an integer $d_{ij}$ such that for all $v \in V_i$, it holds
\begin{equation*}
d_{ij}=deg(v,V_j):=\# \left\{e \in E : e=\left\{v,w\right\}, w \in V_j \right\}
\end{equation*}
independently of $v \in V_i$.
\end{definition}

\rev{The difference between equitable and almost equitable partitions is that, in the former case, subgraph $G_i$ of $G$ induced by $V_i$ has regular structure, whereas the latter definition does not impose any structural condition into $G_i$.}

Ideally we can think of a network $\tilde{G}$ whose node set has an almost equitable partition as a network $G$ with equitable partition where links between nodes in one or more communities have been added or removed.

\rev{The objective is to obtain lower bounds on threshold $\tau_c^{(1)}$, useful in determining a safety region for the extinction of epidemics. We start assuming  that links are added only.}

To this aim, let us consider two graphs $G=(V,E)$ and $\tilde{G}=(V,\tilde{E})$ with the same partition $\{V_1, \dots, V_n\}$, but different edge sets $E\varsubsetneq \tilde{E}$, and assume $G$ to have an equitable partition but $\tilde{G}$ to have merely an almost equitable partition. Then if $\tilde{A}$ and $A$ are the adjacency matrices of $\tilde{G}$ and $G$ respectively it holds
\[
\tilde{A} = A + R,
\]
where $R = \diag(R_1, \dots, R_n)$; the dimension of $R_i$ is $k_i \times k_i$ for $i=1,...,n$, as before $k_i$ is the order of $G_i$ and $n$ is the number of the communities. 

The theorem of Weyl can be applied to $\tilde{A}=A+R$ and then it yields
\begin{equation}\label{IneqA}
\lambda_1(\tilde{A}) \leq \lambda_1(A) + \lambda_1(R).
\end{equation}
\rev{In the following we shall provide a more explicit formulation of the right hand side of \eqref{IneqA} involving the number of added edges.}
\begin{figure}[t]	
\centering
\includegraphics[width=0.55\textwidth]{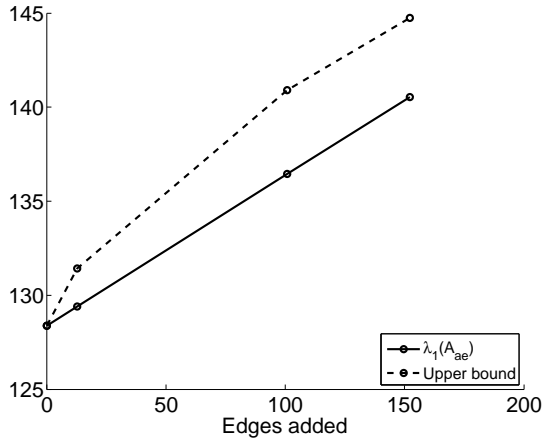}
\caption{Comparison of the bound and the spectral radius for a $40$-communities network. Each community has $k=25$ nodes, whose 
internal structure is a initially a ring; the perturbation graph is obtained by adding in each of them the same increasing number of links. The spectral radius of the adjacency matrix $\tilde{A}$ \rev{($A_{ae}$ in the legend, where the subscript ``ae'' stays for ``almost equitable'')} is compared to the upper bound as a function of the links added in each community.} \label{fig:fig4}
\end{figure}

\begin{proposition}\label{propprop}
Let $G=(V,E)$ and $\tilde{G}=(V,\tilde{E})$ be two graphs 
and consider a partition $\{V_1, \dots, V_n\}$ of the set of vertices $V$; we shall denote by  $G_i=(V_i,E_i)$ and  $\tilde{G_i}=(V_i,\tilde{E_i})$ the subgraph of $G$ and $\tilde{G}$ induced by the cell $V_i$, respectively, for $i=1,...n$. Assume this partition to be equitable for $G$ and almost equitable for $\tilde{G}$. Let $E\subset \tilde{E}$ with
\[
\tilde{E}\setminus E=\bigcup_{i=1}^n (\tilde{E_i}\setminus E_i)
\] 
(i.e., the edge sets can only differ within cells) and denote by $R$ the adjacency matrix corresponding to a graph with $\tilde{E}\setminus E$ as edge set. Finally, let us denote by $G_i^C$ the graph with edge set $\tilde{E_i}\setminus E_i$ and whose node set is simply the set of endpoints of its edges (i.e., no further isolated nodes).
\begin{enumerate}
\item If $\Delta(G_i^C)$ denotes the maximal degree in $G_i^C$, $i=1, \dots, n$,
then
\begin{equation*}
\lambda_1(R)\leq \max_{1\le i\le  n} \min \left\{\sqrt{\frac{2e_i(k_i-1)}{k_i}}, \Delta(G_i^C)\right\}\ ,
\end{equation*}
 where $e_i$ is the number of edges added to $G_i$, i.e., $e_i = (|\tilde E_i| - |E_i|)$, and $k_i$ is the number of nodes in $V_i$.
\item If additionally $G^C_i$ is connected for each $i=1, \dots, n$, then
\begin{equation*}
\lambda_1(R)\leq \max_{1\le i\le n} \min \left\{\sqrt{2e_i-k'_i+1}, \Delta(G_i^C)\right\}\ ,
\end{equation*}
\end{enumerate}
 where $k'_i$ is the number of nodes of $G_i^C$.
\end{proposition}
\begin{proof}
(1) By assumption, $R$ is a diagonal block matrix whose blocks $R_i$ are the adjacency matrices of the induced subgraphs $G_i^C$. Thus, $\lambda_1(R)$ is the maximum of all spectral radii $\lambda_1(R_i)$. On the other hand, one has by~\cite[(3.45)]{Graph} that \begin{equation*}
\lambda_1(R_i) \leq \min \left\{\sqrt{\frac{2e_i(k_i-1)}{k_i}}, \Delta(G_i^C)\right\}.
\end{equation*}
and the claim follows.\\
(2) By Gershgorin's theorem, the spectral radius of an adjacency matrix \rev{of a graph without loops} is never larger than the graph's maximal degree, i.e., $\lambda_1(R_i)\le \Delta(G_i^C)$.
By assumption, there exists a permutation of the vertices in $V_i$ such that the matrix $R_i$ has the form
\[
R_i=
\begin{bmatrix}
 R'_{i} & \textbf{0} \\ 
  \textbf{0} & \textbf{0} \\ 
\end{bmatrix}
\]
where $R'_i$ is the adjacency matrix of a connected graph with $k_i'$ nodes
(i.e., the block $R'_i$ has dimension $k'_i \times k'_i$).
Now, we deduce from~\cite[art.~50]{Graph} that
\begin{equation*}
\lambda_1(R'_i) \leq \sqrt{2e_i -k'_i-1},
\end{equation*}
and since $\lambda(R_i)=\lambda(R'_i)$, the statement follows.
\end{proof}


By using estimate \eqref{lowbound} and Proposition~\ref{propprop} in the first and the second term on the right hand side of~\eqref{IneqA}, respectively, we deduce
\begin{equation}\label{upperAlmost}
\lambda_1(\tilde{A}) \leq 
\max_{1\le i\le n} \lambda_1(C_{V_i}) 
+ \lambda_1(\hB) +  \max_{1\le i\le n} \min \left\{\sqrt{\frac{2e_i(k_i-1)}{k_i}}, \Delta(G_i^C)\right\}.
\end{equation}
The inequality in (\ref{upperAlmost}) gives us a lower bound for the epidemic threshold in the case of a graph whose partition of nodes set is almost equitable. Actually
\begin{equation}\label{boundAE}
\tau_c^{(1)}=\frac{1}{\lambda_1(\tilde{A})} \ge \tau^\star =
\frac{1}{\max\limits_{1\le i\le n} \lambda_1(C_{V_i}) + \lambda_1(\hB)+  \max\limits_{1\le i\le n} \min \left\{\sqrt{\frac{2e_i(k_i-1)}{k_i}}, \Delta(G_i^C)\right\}}.
\end{equation}

Now let us consider the case where we remove edges, inside the communities, in a network whose set nodes has an equitable partition, thus because the spectral radius of an adjacency matrix is monotonically non increasing under the deletion of edges, we have 
$$\lambda_1(\tilde{A}) \leq \lambda_1(A)$$
whence 
\begin{equation*}\label{boundAEdel}
\frac{1}{\lambda_1(\tilde{A})} \geq \frac{1}{\lambda_1(A)} \geq  \min_{i} \frac{1}{d_{ii} + \lambda_1(\hB)}.
\end{equation*}

The bounds developed so far support the design of community networks with safety region for the effective spreading rate, that guarantees the extinction of
epidemics.
E.g. if we consider some  $G_i$, $i=1,\ldots,n$, it is possible to connect them such in a way to form a graph $\tilde{G}=(V,\tilde{E})$ with an almost equitable partition. Now, any subgraph obtained from $\tilde{G}$, by removing edges inside the communities, will have smaller spectral radius than $\tilde{G}$, and consequently a larger epidemic threshold. Thus the lower bound in \eqref{boundAE} still holds.

\section{Conclusion}

\rev{In this work we have discussed the relation between the epidemic threshold of a given graph with equitable partitioning of its node set, and the spectral properties of the corresponding quotient matrix.} Because the quotient matrix $Q$ has the same spectral radius of $A$, \rev{this may lead to a significative computational advantage in the calculation} of $\lambda_1(A)$ and, consequently, of $\tau_c^{(1)}$, since the order of $Q$ is smaller than that of $A$. 

\rev{A novel expression has been derived for the lower bound on $\tau_c^{(1)}$ as function of network metrics, e.g., the maximum among the internal degrees of the nodes over all communities. In practice this value can be adopted to determine a safety region for the extinction of epidemics, i.e., by forcing the effective spreading rate below the lower bound; it can be also useful in order to design new network architectures  
robust to long-term, massive infections.} 

\rev{In the analysis, we have showed} that it is possible to reduce the number of equations representing the time-change of infection probabilities using the quotient matrix $Q$, when all nodes belonging to the same community have the same initial conditions. After proving the existence of a positively invariant set for the original system of $N$ differential equations, we have shown that the non-zero steady-state infection probabilities \rev{belongs to this invariant set, and that it can be computed by the reduced system of $n$ equations}.

Finally, we have also considered the case when the partition is almost equitable. \rev{An input graph whose partition is equitable can be perturbed, by adding or removing edges inside communities, in order to obtain a graph with an almost equitable partition. A lower bound for the epidemic threshold
has been derived, and the effect of perturbations of the communities' structure has been explored.}


\subsection*{Acknowledgments}
The authors would like to thank Piet Van Mieghem for providing interesting and useful comments to an 
early draft of this work. \rev{The helpful comments of two anonymous reviewers are also 
gratefully acknowledged.}



\begin{thebibliography}{}

\bibitem{Hanski}
{\sc I. Hanski and O. Ovaskainen} 
{ \em Metapopulation theory for fragmented landscapes},
Theoretical population biology, 64(1), 119-127, 2003.


\bibitem{Masuda2010}
{\sc N. Masuda}
{\em Effects of diffusion rates on epidemic spreads in metapopulation networks},
{PLoS ONE}, 5(3), 2010.


\bibitem{Allen2007} 
{\sc L. J. S. Allen, B. M. Bolker, Y. Lou and A. L. Nevai  }, 
{\em Asymptotic profiles of the steady states for an SIS epidemic patch model}, 
SIAM Journal on Applied Mathematics, 67(5), 1283-1309, 2007.



\bibitem{Colizza2008}
{\sc V.~Colizza and A.~Vespignani},
\newblock {\em Epidemic modeling in metapopulation systems with heterogeneous
  coupling pattern: theory and simulations.}
\newblock {J. Theoret. Biol.}, 251(3):450--467, 2008.




\bibitem{Poletto2013}
{\sc C.~Poletto, S.~Meloni, V.~Colizza, Y.~Moreno, and A.~Vespignani}, 
{\em Host mobility drives pathogen competition in spatially structured
  populations}.
\newblock {PLoS Computational Biology}, 9(8):e1003169, 2013.

\bibitem {VanMieghem2009}
{\sc P. Van Mieghem, J. Omic and R. E. Kooij},
{\em Virus Spread in Networks},
IEEE/ACM Transaction on Networking, 17:1--14, 2009.


				
\bibitem{NonM} {\sc P. Van Mieghem and R. van de Bovenkamp}, 
{\em Non-Markovian infection spread dramatically alters the SIS epidemic
threshold in networks}, Physical Review Letters, 110(10):108701, 2013.

\bibitem {genInf} 
{\sc E. Cator, R. van de Bovenkamp, and P. Van Mieghem}, 
{\em Susceptible-Infected-Susceptible epidemics on networks with
general infection and curing times},
 Physical Review E, 87(6):062816, 2013.

\bibitem{Bonaccorsi}
{\sc S.\ Bonaccorsi, S.\ Ottaviano, F.\ De Pellegrini, A.\ Socievole and P.\ Van Mieghem}, 
{\em Epidemic outbreaks in two-scale community networks},
 Physical Review E, 90(1): 012810, 2014.	

\bibitem{Ball1997}
{\sc F.~Ball, D.~Mollison, and G.~Scalia-Tomba},
{\em Epidemics with two levels of mixing},
 The Annals of Applied Probability, 7(1):46--89, 1997.

\bibitem{Ball2008}
{\sc F. G. Ball and P. J. Neal},
{\em Network epidemic models with two levels of mixing},
Math. Biosci, 212:69-87, 2008. 

\bibitem{Ross2010}
{\sc J.~V. Ross, T.~House, and M.~J.~Keeling},
\newblock Calculation of disease dynamics in a population of households.
\newblock {\em PLoS ONE}, 5(3), 2010.






									
\bibitem{VanMieghem2012b}
{\sc P.~Van~Mieghem and E.~Cator},
{\em Epidemics in networks with nodal self-infection and the epidemic threshold},
Physical Review E , 86(1):016116, 2012. 
				
\bibitem{VanMieghem2012a}
{\sc P.~Van~Mieghem},
{\em The viral conductance of a network},
Computer Communications, 35(12):1494--1506, 2012.
														
\bibitem{VanMieghem2011}
{\sc P.~Van~Mieghem}, 
{\em  The $N$-Intertwined SIS epidemic network model},
 Springer Computing, 93(2):147--169, 2011.
								
\bibitem{Pollett90}								
{\sc P. K Pollett and A. J. Roberts},
{\em A description of the long-term behaviour of absorbing continuous-time Markov chains using a centre manifold},
Advances in applied probability, 111-128, 1990.								


\bibitem{NasselCLosed}
{\sc{I. Nasell}}, 
{\em The quasi-stationary distribution of the closed endemic SIS model},
 Advances in Applied Probability,  895-932, 1996.			


\bibitem{Nassell2002}
{\sc I. Nasell},
{\em Stochastic models of some endemic infections},
 Mathematical biosciences, 179.1 : 1-19, 2002.																		
																													

 
\bibitem{Bailey1975}
{\sc N.T.J.~Bailey},
{\em The mathematical theory of infectious diseases and its applications},
 Charles Griffin, London, 1975.


\bibitem{Daley1999}
{\sc D.J.~Daley, and J.~Gani},
{\em Epidemic Modelling: An Introduction},
 Cambridge Univ. Press, Cambridge, 1999.


\bibitem{Pastor2001}
{\sc R.~Pastor-Satorras, and A.~Vespignani},
{\em Epidemic spreading in scale free networks},
 Physical Review Letters, 86(14): 3200, 2001.

\bibitem{VanMieghem2013}
{\sc P.~Van~Mieghem},
{\em Decay towards the overall-healthy state in SIS epidemics on networks},
arXiv:1310.3980, 2014

\bibitem{Pollett}	
{\sc P. K. Pollett},
{\em Quasi-stationary distributions: a bibliography},
\url{http://www.maths.uq.edu.au/~pkp/papers/qsds/qsds.pdf}, 2008. 

\bibitem{Draief2010}
{\sc M.~Draief and L.~Massouli\'e},
 {\em Epidemics and Rumours in Complex Networks},
London Math.\ Society Lecture Node Ser. 369, Cambridge University Press, Cambridge, 2010.

\bibitem{Wang2003}
{\sc Y.~Wang, D. ~Chakrabarti, C.~ Wang, and C. Faloutsos},
{\em Epidemic spreading in real networks: an eigenvalue viewpoint},
Proc.\ 22nd Int.\ Symposium on Reliable Distributed Systems, IEEE, 25-33, 2003

\bibitem{scoglio}
{\sc F. D. Sahneh, C. Scoglio, and P. Van Mieghem}, 
{\em Generalized epidemic mean-field model for spreading processes over multilayer complex networks},
 Networking, IEEE/ACM Transactions on 21.5 :1609-1620, 2013

\bibitem{Cator_positive_correlations} 
{\sc E.\ Cator and P.\ Van Mieghem}, 
{\em Nodal infection in Markovian SIS and SIR epidemics on networks are non-negatively correlated},
Physical Review E, 89(5): 052802, 2014. 

\bibitem{Accuracy}
{\sc P.Van Mieghem,  and R. van de Bovenkamp},
{\em Accuracy criterion for the mean-field approximation in SIS epidemics on networks},
 Physical Review E, 91(3): 032812, 2015.

\bibitem{moreno2014}
{\sc R. J. S{\ a}nchez-Garc{\' i}a, E. C., and Y. Moreno}.
{\em Dimensionality reduction and spectral properties of multilayer networks}
Physical Review E, 89, 052815 (2014).


\bibitem{Schwenk} 
{\sc A.J.\ Schwenk}, 
{\em Computing the characteristic polynomial of a graph},
 Graphs and Combinatorics (Proc.\ Washington D.C.\ 1973), Springer, New York, 153-172, 1974 

\bibitem{VanMieghem2014}
{\sc P.~Van~Mieghem},
{\em Exact Markovian SIR and SIS epidemics on networks and an upper bound for the epidemic threshold},
arXiv:1402.1731, 2014

\bibitem{DiffEqandDynamicalSystem} 
{\sc L. Perko}, 
{\em Differential Equations and Dynamical Systems},
Springer-Verlag, New York, 2001.

\bibitem{Stab}
{\sc  A. Lajmanovich and J. A. Yorke },
{\em A deterministic model for gonorrhea in a nonhomogeneous population},
 Mathematical Biosciences 28.3, 221-236, 1976

\bibitem{MatAn}
{\sc R.A. Horn and C.R. Johnson},
{\em Topics in Matrix Analysis}
SIAM, 1990.
													
\bibitem{godsil} {\sc C.D. Godsil and B.D. McKay}.
{\em Feasibility conditions for the existence of walk-regular graphs}, 
Linear Algebra and its Applications,
30: 15-61, 1980.
			


\bibitem{Graph}
{\sc P.~Van~Mieghem}.
{\em Graph Spectra for Complex Networks.}
Cambridge University Press, Cambridge, 2011.

\end{thebibliography}
\end{document}